\documentclass[runningheads]{llncs}

\usepackage{xifthen}
\newboolean{arxiv}
\setboolean{arxiv}{true}

\usepackage[linesnumbered, ruled, vlined]{algorithm2e}

\usepackage{listings}
\usepackage{color}
\usepackage{comment}

\usepackage{booktabs}
\usepackage{multirow}

\usepackage{mathtools}

\usepackage[title]{appendix}

\usepackage{proof}
\usepackage{todonotes}
\usepackage{amsmath}
\usepackage{amssymb}
\usepackage{stmaryrd}
\usepackage{mdframed}
\usepackage{syntax}
\usepackage{caption}
\usepackage{mathrsfs}
\usepackage[numbers,sort,compress]{natbib}
\usepackage{flushend}
\usepackage[hidelinks]{hyperref}
\usepackage{graphicx}
\usepackage{thmtools,thm-restate}

\usepackage[section,numberedsection=autolabel,symbols,nogroupskip,nonumberlist]{glossaries-extra}

\usepackage{rotating}

\ifthenelse{\NOT{\boolean{arxiv}}}{
\usepackage[firstpage]{draftwatermark}
\SetWatermarkText{\hspace*{4.0in}\raisebox{5.75in}{\includegraphics[scale=0.1]{aec-badge-cav}}}
\SetWatermarkAngle{0}
}
{
}

\usetikzlibrary{calc}
\usetikzlibrary{shapes}
\usetikzlibrary{decorations.pathreplacing}
\usetikzlibrary{positioning,chains}

\definecolor{codegreen}{rgb}{0,0.6,0}
\definecolor{codegray}{rgb}{0.5,0.5,0.5}
\definecolor{codepurple}{rgb}{0.58,0,0.82}
\definecolor{backcolour}{rgb}{0.95,0.95,0.92}
\definecolor{assump}{rgb}{0.1,0.95,0.2}
\definecolor{proofob}{rgb}{0.8,0.15,0.2}

\newcommand{\rohit}[1]{\textcolor{red}{#1}}

\newcommand{\uclid}{\textsc{Uclid5}}

\newcommand{\folt}{\mathit{FOL}(\mathcal{T})}
\newcommand{\sigt}{\Sigma_{\mathcal{T}}}
\newcommand{\eqrel}{\simeq{}}

\newcommand{\noteq}{\Delta_{\trv{j}, \trv{k}}}
\newcommand{\myimpl}{\Rightarrow}
\newcommand{\myiff}{\Leftrightarrow}

\newcommand{\unsat}{\ensuremath{\mathsf{unsat}}}
\newcommand{\sat}{\ensuremath{\mathsf{sat}}}

\newcommand{\minit}{\ensuremath{\mathit{Init}}}
\newcommand{\mtr}{\ensuremath{\mathit{Tx}}}

\newcommand{\mname}{M}

\newcommand{\machorig}{{\mname{} = \langle X, \minit{}(X), \mtr{}(X,X') \rangle}}

\newcommand{\valueof}[2]{{#2}({#1})}

\newcommand{\qformulasub}{\forall \trv{0}.~\countq{\trv{1}}{\noteq}{\varphi_{\trv{0}, \trv{1}}}{\triangleleft}{N(Z)}}
\newcommand{\qformula}{\forall \trv{0}.~\countq{\trv{1}}{\noteq}{\varphi}{\triangleleft}{N(Z)}}
\newcommand{\qformulage}{\forall \trv{0}.~\countq{\trv{1}}{\noteq}{\varphi}{\geq}{N(Z)}}
\newcommand{\qformulale}{\forall \trv{0}.~\countq{\trv{1}}{\noteq}{\varphi}{\leq}{N(Z)}}

\newcommand{\qformulavalid}[1]{\forall \trv{0}.~\countq{\trv{1}}{\noteq}{\varphi}{#1}{\countsat{\valid{Y, Z}{}{}}{Y}}}
\newcommand{\qformulavalidprime}[1]{\forall \trv{0}.~\countq{\trv{1}}{\noteq}{\varphi}{#1}{\countsat{\validprime{Y, Z}{}{}}{Y}}}

\newcommand{\fixy}{\mathtt{Y}}
\newcommand{\fixz}{\mathtt{Z}}

\newcommand{\enum}[1]{\mathsf{Y}}
\newcommand{\enumi}[1]{\enum{}[#1]}

\newcommand{\true}{\mathit{true}}
\newcommand{\false}{\mathit{false}}
\newcommand{\obs}[1]{\mathit{obs}(#1)}
\newcommand{\secret}[1]{\mathit{secret}(#1)}

\newcommand{\code}[1]{\texttt{#1}}

\newcommand{\trv}[1]{
  \ifthenelse{\isempty{#1}}
    {\pi}
    {\pi_{#1}}}
\newcommand{\trc}[1]{
  \ifthenelse{\isempty{#1}}
    {\tau}
    {\tau_{#1}}}
\newcommand{\trcj}[1]{
    {\tau_{#1}}}
\newcommand{\trci}[2]{
  \ifthenelse{\isempty{#1}}
    {\tau^{#2}}
    {\tau_{#1}^{#2}}
}
\newcommand{\trcij}[2]{
    {\tau_{#1}^{#2}}
}
\newcommand{\traceset}{\Phi}
\newcommand{\traces}[1]{\traceset{}_{#1}}
\newcommand{\cntv}{\mathcal{C}}
\newcommand{\tracesetcnt}{\traceset{}_{\cntv}}
\newcommand{\tracesetq}{\traceset{}_\varphi}

\newcommand{\countq}[5]{\# #1\!\!:\!#2.~#3 ~#4\,#5}
\newcommand{\countqinst}{\countq{\trv{}}{\noteq}{\varphi}{\triangleleft}{N(Z)}}

\newcommand{\parameter}[1]{
  \ifthenelse{\isempty{#1}}
    {p}
    {p_{#1}}}

\newcommand{\varset}{\mathit{Vars}}
\newcommand{\valid}[2]{
  \ifthenelse{\isempty{#2}}
    {\mathcal{V}
      {
        \ifthenelse{\isempty{#1}}{}{(#1)}
      }
    }
    {\mathcal{V}_{#2}(#1)}
  }

\newcommand{\validprime}[2]{
  \ifthenelse{\isempty{#2}}
    {\mathcal{V'}
      {
        \ifthenelse{\isempty{#1}}{}{(#1)}
      }
    }
    {\mathcal{V'}_{#2}(#1)}
  }

\newcommand{\tren}[1]{
    \ifthenelse{\isempty{#1}}
        {\mathcal{U}}
        {\mathcal{U}(#1)}
    }
\newcommand{\trel}[1]{
    \ifthenelse{\isempty{#1}}
        {\mathcal{E}}
        {\mathcal{E}(#1)}
}

\newcommand{\ifunx}[2]{
    \ifthenelse{\isempty{#2}}
        {\mathscr{I}_{#1}}
        {\mathscr{I}_{#1}(#2)}
}

\newcommand{\sfuny}[2]{
    \ifthenelse{\isempty{#1}}
        {\mathscr{S}_{#1}}
        {\mathscr{S}_{#1}(#2)}
}
\newcommand{\fcnt}[1]{
  \ifthenelse{\isempty{#1}}
  {\mathscr{C}{}}
  {\mathscr{C}(#1)}}

\newcommand{\extn}[1]{
  \ifthenelse{\isempty{#1}}
  {\mathrm{extn}{}}
  {\mathrm{extn}(#1)}}

\newcommand{\stat}[1]{
  \ifthenelse{\isempty{#1}}
    {\sigma}
    {\sigma^{#1}}}

\newcommand{\ltlu}[2]{#1\,\mathsf{\textbf{U}}\,#2}
\newcommand{\ltlg}[1]{\mathsf{\textbf{G}}\,#1}
\newcommand{\ltlx}[1]{\mathsf{\textbf{X}}\,#1}
\newcommand{\ltlf}[1]{\mathsf{\textbf{F}}\,#1}
\newcommand{\predonly}{\mathcal{P}}
\newcommand{\predsub}[1]{\mathcal{P}_{#1}}

\newcommand{\sizeof}[1]{|#1|}

\newcommand{\decr}{\code{decr}}

\newcommand{\countsat}[2]{{\#}{#2}.\,{#1}}
\newcommand{\verf}{\valid{\enum{},\numrnds{}}{f}}
\newcommand{\verfp}{\valid{\enum{},\numrnds{}-1}{f}}
\newcommand{\verfo}{\valid{\enum{},1}{f}}
\newcommand{\vero}{\valid{\enum{},\numrnds{}}{1}}
\newcommand{\ver}{\valid{\enum{},\numrnds{}}{}}
\newcommand{\cverf}{\countsat{\verf{}}{\enum{}}}
\newcommand{\cverfo}{\countsat{\verfo{}}{\enum{}}}
\newcommand{\cverfp}{\countsat{\verfp{}}{\enum{}}}
\newcommand{\cvero}{\countsat{\vero{}}{\enum{}}}
\newcommand{\cver}{\countsat{\ver{}}{\enum{}}}

\newcommand{\challenge}{\mathsf{C}}

\newcommand{\response}{\mathsf{P}}
\newcommand{\responsei}[1]{\mathsf{P}[#1]}
\newcommand{\responseij}[1]{\mathsf{P}_0[#1]}
\newcommand{\responseik}[1]{\mathsf{P}_1[#1]}

\renewcommand{\success}{\mathsf{S}}
\newcommand{\numrnds}{\mathsf{R}}

\newcommand{\numblks}{\mathsf{numBlks}}
\newcommand{\stashsz}{\mathsf{stashSz}}
\newcommand{\pmap}[1]{
  \ifthenelse{\isempty{#1}}
    {\mathsf{pMap}}
    {\mathsf{pMap}_{#1}}}

\newcommand{\request}[1]{
  \ifthenelse{\isempty{#1}}
    {\mathsf{request}}
    {\mathsf{request}_{#1}}}
\newcommand{\leaf}[1]{
  \ifthenelse{\isempty{#1}}
    {\mathsf{leaf}}
    {\mathsf{leaf}_{#1}}}
\newcommand{\remap}{\mathsf{remap}}
\newcommand{\perm}{\mathsf{Y}}
\newcommand{\permhat}{\mathsf{W}}

\newcommand{\dom}[1]{\mathit{dom}(#1)}

\begin{document}
\ifthenelse{\boolean{arxiv}}{
    \title{
        Verification of Quantitative Hyperproperties Using Trace Enumeration Relations
        \thanks{This is an extended version of a paper with the same title that appears at CAV 2020.}
    }
    \author{
        Shubham Sahai\inst{1} \and Pramod Subramanyan\inst{1} \and Rohit Sinha\inst{2}
    }
}
{
    \title{
        Verification of Quantitative Hyperproperties Using Trace Enumeration Relations
    }
    \author{
        Shubham Sahai\inst{1}\orcidID{0000-0002-3434-6937} \and \\
        Pramod Subramanyan\inst{1}\orcidID{0000-0003-2288-3396} \and \\
        Rohit Sinha\inst{2}\orcidID{0000-0001-9107-0239}
    }
}
\titlerunning{Verification of QHPs Using Trace Enumeration Relations}
%

%
\institute{
    Indian Institute of Technology, Kanpur \and
    Visa Research\\
}

\maketitle              

\begin{abstract}
    Many important cryptographic primitives offer probabilistic guarantees of security that can be specified as quantitative hyperproperties; these are specifications that stipulate the existence of a certain number of traces in the system satisfying certain constraints.
    Verification of such hyperproperties is extremely challenging because they involve simultaneous reasoning about an unbounded number of different traces.
    In this paper, we introduce a technique for verification of quantitative hyperproperties based on the notion of trace enumeration relations.
    These relations allow us to reduce the problem of trace-counting into one of model-counting of formulas in first-order logic.
    We also introduce a set of inference rules for machine-checked reasoning about the number of satisfying solutions to first-order formulas (aka model counting).
    Putting these two components together enables semi-automated verification of quantitative hyperproperties on infinite state systems.
    We use our methodology to prove confidentiality of access patterns in Path ORAMs of unbounded size, soundness of a simple interactive zero-knowledge proof protocol as well as other applications of quantitative hyperproperties studied in past work.
\end{abstract}


\section{Introduction}

Recent years have seen significant progress in automated and semi-automated techniques for the verification of security requirements of computer systems~\cite{BartheCsfw04, TerauchiSas05, clarkson-14, hyperltl-15, obsdet-03, CHL16, QCHL17, KoskinenPldi17}. 
Much of this progress has built on the theory of \emph{hyperproperties}~\cite{ClarksonS10}, and
these have been used extensively in analysis of whether systems satisfy secure information flow properties~\cite{AlmeidaUsenix16, Almeida-13,secverilog-15, secverilog-17, hawblitzel-osdi14, tap-ccs17, spectector-18, csf-19, ccs-14, fse-16,fmcad-19} such as observational determinism~\cite{obsdet-03, roscoe-sp-95} and non-interference~\cite{goguen-sp82}. 
Unfortunately, the security specification of several important security primitives cannot be captured by secure information flow properties like observational determinism. 
In particular, observational determinism and non-interference are not applicable when reasoning about algorithms that offer probabilistic -- as opposed to deterministic -- guarantees of confidentiality and integrity. Prominent examples of security primitives offering probabilistic guarantees include Path ORAM~\cite{path-oram-2013} and various zero-knowledge proof protocols.

A promising direction for the verification of such protocols are the class of quantitative hyperproperties~\cite{mcqhyper-2018}, one example of which is deniability~\cite{deny-17,deny-18}. 
Deniability states that for every infinitely-long sequence of observations that an adversary makes, there are (exponentially) many different secrets that could have resulted in exactly these observations. 
Therefore, the adversary learns very little about the secrets in an execution from a particular sequence of observations. 

How does one prove a quantitative hyperproperty like deniability?
Suppose our goal is to show that for every trace of adversary observations, there exist $2^{n}$ traces with the same observations but different secrets. 
Here $n$ is a parameter of the system, e.g., the length of a password in bits.
One option, first suggested by Yasuoka and Terauchi~\cite{yasuoka2014quantitative} and recently revisited by Finkbeiner, Hahn, and Torfah~\cite{mcqhyper-2018}, is to consider the following $k$-trace property, where $k=2^{n}+1$.
\begin{align*}
    \forall & \trv{0}.~\exists \trv{1}, \trv{2}, \dots, \trv{2^{n}}.~ & \\
        & \Big(\bigwedge_{j=1}^{2^{n}} \obs{\trv{0}} = \obs{\trv{j}} \Big)
          \land
          \Big(\bigwedge_{j=1}^{2^{n}} \bigwedge_{k=1}^{2^{n}}~ 
            (j \neq k) \myimpl \secret{\trv{j}} \neq \secret{\trv{k}}\Big) & \nonumber
\end{align*}

The property states that for every trace of the system, there must exist $2^{n}$ other traces with identical observations and pairwise different secrets.
In the above, $\trv{0}, \trv{1}, \dots$ represent trace variables, $\obs{\trv{j}}$ refers to the trace of adversary observations projected from the trace $\trv{j}$, while $\secret{\trv{j}}$ refers to the trace of secret values in the trace $\trv{j}$.
There are at least three problems with the verification of the above property.
First, the size of this property grows exponentially with $n$; verification needs to reason about $2^n$ traces simultaneously and is not scalable.
The second problem is quantifier alternation.
Even if we could somehow reason about $2^n$ traces, we have to show that \emph{for every} trace $\trv{0}$, there \emph{exist} $2^n$ other traces 
satisfying the above condition.
The third problem is that the above technique does not work for \emph{symbolic} bounds. 
While it is possible -- at least in principle -- to use the above construction by picking a specific value of $n$, say $16$, to show that $2^{16}$ traces exist that satisfy deniability, we would like to show that the property holds for all $n$, where $n$ is a state variable or parameter of the transition system.
Capturing the dependence of the trace-count bound on parameters, such as $n$, is important because it shows that the attacker has to work exponentially harder as $n$ increases.
Such general proofs are not possible by reduction to a $k$-trace property because the construction requires $k$ be bounded.

        Recent work by Finkbeiner, Hahn, and Torfah~\cite{mcqhyper-2018} has made significant progress in addressing the first two problems by showing a reduction from $k$-trace property checking into the problem of maximum model counting~\cite{fremont-17}.
        However, their technique still produces a propositional formula whose size grows exponentially in the size of the quantitative hyperproperty.
        Further, model counting itself is a computationally hard problem that is known to be $\#P$-complete, and maximum model counting is even harder.
        As a result, their technique does not scale well and times out on the verification of an 8-bit leakage bound for an 8-bit password.
        Finally, their method does not support symbolic bounds, and therefore cannot be used to verify parametric systems; we verify several examples of such systems in this paper (e.g., Path ORAM~\cite{path-oram-2013} of symbolic size).


    In this  work, we propose a new technique for quantitative hyperproperty verification that addresses each of the above problems.
    Our approach is based on the following insights.
    First, instead of trying to count the number of traces that have the same observations and different inputs, we instead show injectivity/surjectivity from satisfying assignments of a first-order formula to traces of a transition system.
    This allows us to bound the number of traces satisfying the quantitative hyperproperty by the number of satisfying solutions to this formula.
    We introduce the notion of a trace enumeration relation to formalize this relation between the first-order formula and traces of the transition system.
    An important advantage of the above reduction is that proving the validity of a trace enumeration relation is only a hyperproperty -- not a quantitative hyperproperty.

        Next, we develop a novel technique to bound the number of satisfiable solutions to a first-order logic formula, which is of independent interest.
        While this is a hard problem, we exploit the fact that our formulas have a significant amount of structure.
        We introduce a set of inference rules inspired by ideas from  enumerative combinatorics~\cite{wilf-2005,bjorner-2010,zb-2010}.
        These rules allow us to bound the number of satisfying assignments to a formula by making only satisfiability queries.

    In summary, our techniques can prove quantitative hyperproperties with symbolic bounds on parametric infinite-state systems.
        We demonstrate their utility by verifying representative quantitative hyperproperties of diverse applications.

\subsection*{Contributions}
\begin{enumerate}
    \item We introduce a specification language for quantitative hyperproperties (QHPs) over symbolic transition systems and define formal satisfaction semantics for this language.
        Our specification language is more expressive than past work on QHP specification because it allows the bound to be a first-order formula over the state variables of the transition system.
    \item We provide several examples of QHPs relevant to security verification.
        We identify a new class of QHPs, referred to as soundness hyperproperties, applicable to protocols that provide statistical guarantees of integrity.
    \item We propose a novel semi-automated verification methodology for proving that a system satisfies a QHP.
        Our  methodology applies to properties that involve a single instance of quantifier alternation and works by reducing the problem of QHP verification to that of checking non-quantitative hyperproperties over two and three traces of the system and counting satisfiable solutions to a formula in first-order logic.
    \item We introduce a set of inference rules for bounding the number of satisfiable solutions to a first-order logic formula, using only satisfiability queries.

    \item We demonstrate the applicability of our specification language  and verification methodology by providing proofs of security for Path ORAM, soundness of a simple zero-knowledge protocol, as well as examples taken from prior work on quantitative security specifications.
        We show that our verification methodology scales to larger systems than could be handled in prior work.
        To the best of our knowledge, our work is the first machine-checked proof of confidentiality of the access patterns in Path ORAM.
\end{enumerate}

\section{Motivating Example}
\label{sec:motiv}
  
In this section, we first introduce the model of transition systems used in this paper.
We then discuss quantitative hyperproperty (QHP) specification and verification for our running example -- a simple zero-knowledge puzzle.

\subsection{Preliminaries}

Let $\folt$ denote first-order logic modulo a theory $\mathcal{T}$. 
The theory $\mathcal{T}$ is assumed to be multi-sorted, includes the theory of linear integer arithmetic (LIA), and contains the $=$ relation.
Let $\sigt$ be the theory $\mathcal{T}$'s signature: the set consisting of the constant, function, and predicate symbols in the theory. 
We say that a formula is a $\sigt$-formula if it consists of the symbols in $\sigt$ along with variables, logical connectives, and quantifiers.
We only consider theories which are such that the set of satisfying assignments for any $\sigt$-formula is a countable set.\footnote{Our experiments mostly use the AUFLIA theory which allows arrays, uninterpreted functions, and linear integer arithmetic.}

For every variable $x$, we will assume there exists a unique variable $x'$, which we refer to as the primed version of $x$.
We will use $X$, $Y$, and $Z$ to denote sets of variables.
Given a set of variables $X$, we will use $X'$ to refer to the set consisting of the primed version of each variable in $X$, that is $X' = \{ x'~|~x \in X \}$.
Similarly $X_1$, $X_2$, etc. are sets consisting of new variables defined as follows: $X_1 = \{ x_1~|~x \in X \}$ and $X_2 = \{ x_2~|~x \in X \}$.
We will use $F(X)$ to denote the application of a function or predicate symbol $F$ on the variables in the set $X$.
A satisfying assignment $\sigma$ to the formula $F(X)$ is written as $\sigma \models F(X)$. 
Given a formula $F(X)$ and a satisfying assignment $\sigma$ to this formula, we will denote the valuation of the variable $x \in X$ in the assignment $\sigma$ as $\valueof{x}{\sigma}$.
We will abuse notation in two ways and also write $\valueof{X}{\sigma}$ to refer to a map from the variables $x \in X$ to their assignments in $\sigma$.
We will also write $\valueof{G(X)}{\sigma}$ to denote the valuation of the term $G(X)$ under the assignment $\sigma$.

The number of satisfiable assignments for the variables in the set $X$ to a formula $F(X, Y)$ as a function of the variables $Y$ will be denoted by $\countsat{F(X, Y)}{X}$. $\countsat{F(X, Y)}{X}$ is the function $\lambda \mathrm{Y}~.~ |\{ \valueof{X}{\sigma} ~|~ \sigma \models F(X, \mathrm{Y}) \}|$ evaluated at $Y$; $|S|$ is the cardinality of the set $S$.
For example, consider the predicate $f(i, n) \doteq (0 \leq i < 2n)$.
In this case, $\countsat{f(i, n)}{i} = \max{(0, 2n)}$, meaning that for a given value of $n > 0$, there are $2n$ satisfying assignments to $i$.


\begin{definition}[Transition System]\label{defn:symb-TS}
    A transition system $M$ is defined as the tuple $\machorig{}$. $X$ is a finite set of (uninterpreted) constants that represents the state variables of the transition system.
$\minit$ and $\mtr$ are $\sigt$-formulas representing the initial states and the transition relation, respectively.
$\minit$ is defined over the signature $\sigt \cup X$. $\mtr$ is over the signature $\sigt \cup X\cup X'$; $X$ represents the pre-state of the transition and $X'$ represents its post-state.
\end{definition}

A state of the system is an assignment to the variables in $X$.
We use $\stat{0}, \stat{1}, \stat{2}$ etc. to represent states.
A trace of the system $M$ is an infinite sequence of states $\trc{} = \stat{0}\stat{1}\stat{2}\dots$ $\stat{i}$ $\dots$ such that $\minit(\stat{0})$ is valid and for all $i \geq 0$, $\mtr(\stat{i}, \stat{i+1})$ is valid; in order to keep notation uncluttered, we will often drop the $\geq 0$ qualifier when referring to trace indices.
We assume that every state of the transition system has a successor: for all $\stat{}$ there exists some $\stat{}'$ such that $\mtr(\stat{}, \stat{}')$ is valid,
ensuring every run of the system is infinite.
We will represent traces by $\tau, \tau_1, \tau_2$, etc.
Given a trace $\trc{}$, we refer to its $i^{th}$ element by $\tau^i$.
If $\trc{} = \stat{0}\stat{1}\dots$, then $\trci{}{0} = \stat{0}$ and $\trci{}{1} = \stat{1}$.
The notation $\tau^{[i,\infty]}$ refers to the suffix of trace  $\trc{}$ starting at index $i$.
The set of all traces of the system $M$ is denoted by $\traces{M}$.  
Given a state $\stat{}$ and a variable $x \in X$, $\valueof{x}{\stat{}}$ is the valuation of $x$ in the state $\stat{}$.

\subsection{Motivating Example: Zero-Knowledge Hats}

\label{sec:motivating-example} 
Zero-knowledge (Z-K) proofs are constructions involving two parties: \emph{a prover} and  \emph{a verifier}, where the prover's goal is to convince the verifier about the veracity of a given statement without revealing any additional information. 
We motivate the need for quantitative hyperproperty verification using a Z-K puzzle.

\subsubsection{Puzzle Overview:}
  Consider the following scenario. Peggy has a pair of otherwise identical hats of different colors (say, yellow and green). 
  She wants to convince Victor, who is yellow-green color blind, that the hats are of different colors, without revealing the colors of the hats. 
  This problem can be solved using the following interactive protocol. 
  Peggy gives both hats to Victor, and Victor randomly chooses a hat behind a curtain and shows it to Peggy.
  Next, he goes back behind the curtain and uniformly randomly chooses if he wants to switch the hat or not.
  He now appears in front of Peggy and asks: ``Did I switch?''

  If the hats are really of different colors, Peggy will be able to answer correctly with probability 1. 
  If Peggy is cheating -- the hats are in fact of the same color -- her best strategy is to guess, and with probability $0.5$ she will answer incorrectly. 
  If the interaction is repeated $k$-times, Peggy will be caught with probability $1 - 2^{-k}$. 
  The interaction between Peggy and Victor only reveals the fact that Peggy can detect a switch and not the color of the hat, making this zero-knowledge.  
\subsubsection{Verification Objectives:}
  A zero-knowledge proof must satisfy three properties: \emph{completeness} (an honest prover should be able to convince an honest verifier of a true statement), \emph{soundness} (a cheating prover can convince an honest verifier with negligible probability) and \emph{zero-knowledge} (no information apart from the veracity of the statement should be revealed). 
  Completeness is a standard trace property, while zero-knowledge is the 2-safety property of indistinguishability.
  Consequently, the main challenge in automated verification of the zero-knowledge protocol described above is that of soundness.
  In this section, we discuss its specification and verification using quantitative hyperproperties.

  \begin{figure}[htbp]
    \begin{mdframed}[innerrightmargin=0pt, innerleftmargin=3pt,innertopmargin=-5pt,innerbottommargin=3pt]
        \begin{align*}
        &X          &\doteq~ & \{\challenge, \response, \success, i, \numrnds\}           \\
        &\minit(X)  &\doteq~ & (\forall i.~0 \leq \challenge[i] \leq 1) \land
                               (\forall i.~0 \leq \response[i] \leq 1) \land
                                \success \land (i = 1) \land (\numrnds > 0)                \\
        &\mtr(X,X') &\doteq~ & (\challenge' = \challenge) \land (\response' = \response)  \land (\numrnds{}' = \numrnds{}) \land
                                \big(\success' = \big( \success \land (\challenge[i] = \response[i])\big) \big) ~\land \\
        &           &        &  i' = \min{(i+1,\numrnds{})}~
      \end{align*}
    \end{mdframed}
    \caption{Transition system model of the example protocol.}
    \label{fig:TS-hat}
    \ifthenelse{\NOT{\boolean{arxiv}}}{
      \vspace{-0.2in} 
    }{}
  \end{figure}
  \subsubsection{Soundness as a Quantitative Hyperproperty:}
  \label{sec:ex-prop-desc}
  Consider the transition system  $\machorig{}$, shown in Figure \ref{fig:TS-hat}, representing this protocol.
  The variable $\numrnds{}$ is a \emph{parameter} of the system and refers to the number of rounds of the protocol.
  $\challenge$ and $\response$ are boolean arrays representing the challenges from the verifier to the prover, and the responses from the prover to the verifier, respectively.  $i$ is the current round, and $\success{}$ is a boolean flag which corresponds to whether the zero-knowledge proof has succeeded.
  $\challenge$ and $\response$ are initialized non-deterministically to model the fact that the verifier chooses their challenges randomly, and a cheating prover's best strategy is guessing. 
    While a cheating prover can use any strategy, if the challenges are indistinguishable to her, then the best strategy is to sample responses from a uniform distribution.

  Soundness is captured by the following quantitative hyperproperty (QHP):
  \begin{align}
      \forall \trv{0}.\countq{\trv{1}}{\ltlf{(\delta_{\trv{j}, \trv{k}})}}{\ltlg{(\psi_{\trv{0}, \trv{1}})}}{\geq}{2^\numrnds - 1}
    \label{eqn:motiv-ex}
  \end{align}
  We will provide formal satisfaction semantics for QHPs in Section~\ref{sec:prelim}. 
  For now, we informally describe its meaning.
  The term $\countq{\trv{1}}{\ltlf{(\delta_{\trv{j}, \trv{k}})}}{\ltlg{(\psi_{\trv{0}, \trv{1}})}}{\geq}{2^\numrnds{}-1}$ introduces a counting quantifier which stipulates the existence of at least $2^\numrnds - 1$ traces satisfying certain conditions:
  (i) these traces must all be pairwise-different, where difference is defined by satisfaction of the formula $\ltlf{(\delta_{\trv{j}, \trv{k}})}$ and (ii)
  all of these traces must be related to trace $\trv{0}$ by the relation $\ltlg{(\psi_{\trv{0}, \trv{1}})}$. 

  The state predicates $\delta$ and $\psi$ are defined as follows.
  \begin{align*}
      \delta(\stat{}_1, \stat{}_2) ~\doteq~ &
          \valueof{\responsei{i}}{\stat{}_1} \neq
          \valueof{\responsei{i}}{\stat{}_2} & \nonumber \\
      \psi(\stat{}_1, \stat{}_2) ~\doteq~ & 
            \big(\valueof{(i = \numrnds) \myimpl \success}{\stat{}_1} \myimpl
            \valueof{(i = \numrnds) \myimpl \lnot\success}{\stat{}_2}\big) 
            ~~~\land\\
            & 
            \big(\valueof{\challenge{}}{\stat{}_1}= 
            \valueof{\challenge{}}{\stat{}_2} \land 
            \valueof{\numrnds{}}{\stat{}_1} = \valueof{\numrnds{}}{\stat{}_2} \big)& \nonumber
  \end{align*}

  The requirement imposed by $\delta$ is that Peggy's responses be different at some step $i$ for every pair of traces captured by the counting quantifier.  $\psi$ says that if trace $\trv{0}$ is a trace where Peggy's cheating succeeds (i.e., $\success{} = \true$ when $i = \numrnds{}$), then in all traces captured by $\trv{1}$, the challenges and number of rounds are the same as $\trv{0}$ but Peggy's cheating is detected by Victor (i.e., $\success = \false$ when $i = \numrnds{}$).
  These requirements are illustrated in Figure~\ref{fig:tep-qhp}(b).
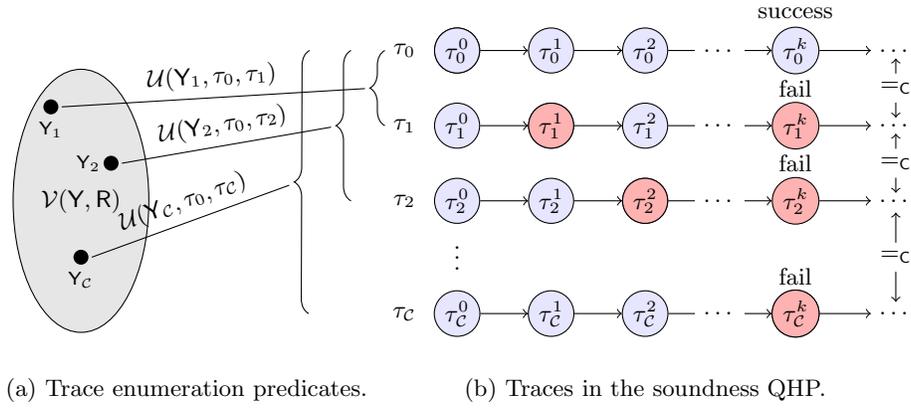
\begin{figure}[!htb]
      \begin{center}
    \begin{tikzpicture}
      \newcommand{\pictext}[1]{\footnotesize{#1}}
      \newcommand{\smalltxt}[1]{\scriptsize{#1}}
      \tikzstyle{state}=[circle,minimum width=0.5cm,inner sep=0.05cm,
                         draw,node distance=1.25cm,fill=green!05];
      \foreach \s/\r/\lb/\pos/\col/\fcol/\scol/\lab in {
          1/0/0/0/blue!10/blue!10/blue!10/success, 
          2/-1/1/1/blue!10/red!30/red!30/fail, 
          3/-2/2/2/blue!10/red!30/red!30/fail, 
          n/-3.5/\cntv{}/k/blue!10/red!30/red!30/fail
      }
      {
          \node[state,fill=\col] (s0\s) at (0,\r) 
             {\pictext{$\trcij{\lb}{0}$}};
          \node[left of=s0\s,node distance=0.7cm] (t\s) {\pictext{$\trc{\lb}$}};
          \node[state, fill=\col, right of=s0\s] (s1\s)
             {\pictext{$\trcij{\lb}{1}$}};
          \node[state, fill=\col, right of=s1\s] (s2\s)
             {\pictext{$\trcij{\lb}{2}$}};
          \node[state, fill=\scol, right of=s2\s, node distance=2cm] (sk\s)
             {\pictext{$\trcij{\lb}{k}$}};
          \node[above of=sk\s,node distance=0.5cm] (l\s) {\pictext{\lab}};
          
          \node[state,fill=\fcol] at (s\pos\s) {\pictext{$\trcij{\lb}{\pos}$}};

         \draw[->] (s0\s) -- (s1\s);
         \draw[->] (s1\s) -- (s2\s);
         \draw[->] (s2\s) -- coordinate[midway](sj\s) 
                   node[fill=white] {\pictext{$\dots$}} (sk\s);
         \draw[->] (sk\s) -- ($(sk\s) + (1cm,0)$) node[right] 
                   (sl\s) {\pictext{$\dots$}};

      }
      \path (sl1) -- node (eq1) {\pictext{$=_{\challenge{}}$}} (sl2);
      \draw[->] (eq1.north) -- ($(eq1.north) + (0,0.2)$);
      \draw[->] (eq1.south) -- ($(eq1.south) + (0,-0.2)$);

      \path (sl2) -- node (eq2) {\pictext{$=_{\challenge{}}$}} (sl3);
      \draw[->] (eq2.north) -- ($(eq2.north) + (0,0.2)$);
      \draw[->] (eq2.south) -- ($(eq2.south) + (0,-0.2)$);

      \path (sl3) -- node (eq3) {\pictext{$=_{\challenge{}}$}} (sln);
      \draw[->] (eq3.north) -- ($(eq3.north) + (0,0.4)$);
      \draw[->] (eq3.south) -- ($(eq3.south) + (0,-0.4)$);

      \draw[draw=white] (s03) -- coordinate[midway] node {\pictext{$\vdots$}} 
                        ($(s0n) + (-0,0.5cm)$);

      \node[below,yshift=-0.75cm] (label1) at (s2n) 
        {(b) Traces in the soundness QHP.};

      \node[left of=label1,xshift=-5cm,text width=5cm] {(a) Trace enumeration predicates.};

      \node[ellipse,draw,fill=gray!20, 
            minimum width=1.5cm, minimum height=3.5cm] (voval) at (-5,-2) 
            {$\valid{\enum{}, \numrnds{}}{}$};
      \draw[decorate,decoration={brace,amplitude=2mm,mirror}] 
        (t1.west) -- 
        node[xshift=-0.15cm] (mid1) {}
        (t2.west);
      \draw[decorate,decoration={brace,amplitude=2mm,mirror}] 
        ($(t1.west) + (-0.45,0)$) -- 
        node[xshift=-0.15cm] (mid2) {}
        ($(t3.west) + (-0.45,0)$);
      \draw[decorate,decoration={brace,amplitude=2mm,mirror}] 
        ($(t1.west) + (-1.0,0)$) -- 
        node[xshift=-0.15cm] (mid3) {}
        ($(tn.west) + (-1.0,0)$);
      \node[circle,fill,inner sep=2pt] (sat1) at ($(voval) + (-0.4,1.25)$) {};
      \node[below of=sat1,node distance=0.3cm] {\smalltxt{$\enum{}_1$}};
      \node[circle,fill,inner sep=2pt] (sat2) at ($(voval) + (0.4,0.5)$) {};
      \node[left of=sat2,node distance=0.3cm] {\smalltxt{$\enum{}_2$}};
      \node[circle,fill,inner sep=2pt] (sat3) at ($(voval) + (0,-0.75)$) {};
      \node[below of=sat3,node distance=0.3cm] {\smalltxt{$\enum{}_{\cntv{}}$}};
      \draw[-] ($(sat1.west) + (0.25,0)$) -- 
               node[midway, above, sloped] 
               {$\tren{\enum{}_1, \trc{0}, \trc{1}}$} 
               (mid1);
      \draw[-] ($(sat2.west) + (0.25,0)$) -- 
               node[midway, above, sloped] 
               {$\tren{\enum{}_2, \trc{0}, \trc{2}}$} 
               (mid2);
      \draw[-] ($(sat3.west) + (0.25,0)$) -- 
               node[midway, above, sloped] 
               {$\tren{\enum{}_{\cntv{}}, \trc{0}, \trc{\cntv{}}}$} 
               (mid3);
    \end{tikzpicture}
  \end{center}
    \caption{Using trace enumeration predicates to verify the soundness QHP.}
    \label{fig:tep-qhp}
\end{figure}

  The QHP requires that for every trace in which a cheating prover succeeds in tricking the verifier for a given trace of challenges, there are $2^\numrnds{} - 1$ other traces with the same challenges in which the prover's cheating is detected.
  Even though soundness is a probabilistic property over the distribution of the system's traces, it can be reduced to counting (and thus specified as a QHP) because each execution trace is sampled uniformly from a finite set.
  Therefore, if the QHP is satisfied, Peggy's probability of successful cheating is upper-bounded by $2^{-\numrnds{}}$.

\subsection{Solution Outline}
\ifthenelse{\boolean{arxiv}}{
\begin{figure}[!h!tb]
    \begin{center}
        \begin{tikzpicture}
            \tikzstyle{blocknode}=[ draw, rectangle, align=center, text width=6.0cm, minimum height=0.7cm, font=\footnotesize\sffamily, inner sep=0.5ex, fill=blue!05, 
            rounded corners];
            \tikzstyle{label}=[draw, rectangle, align=center, rounded corners, font=\small\bf, inner sep=1ex, fill=cyan!30, 
                                minimum width=3.2cm, minimum height=0.6cm];
            \tikzstyle{textnode}=[text width=3cm, font=\footnotesize\it];

            \tikzstyle{arrow}=[very thick,->,>=stealth];
            \tikzstyle{arrowtextu}=[text width=2cm, above, align=center]
            \tikzstyle{arrowtextd}=[text width=2cm, below, align=center, font=\footnotesize\it]
            
            \begin{scope}[node distance=0.6cm, start chain=1 going below, every join/.style=arrow, xshift=-2cm,]
                \coordinate[on chain=1] (tc);
                \node[blocknode, on chain=1] (n2)
                    {Construct Transition System Model};
                \node[blocknode, join, on chain=1] (n3)
                    {Construct Enumeration: $\valid{}{}$, $\tren{}$};
                \coordinate[on chain=1] (cc);
                \node[blocknode, on chain=1, yshift=-1.6cm, fill=yellow!30] (n7)
                    {Counting};
            \end{scope}

            \begin{scope}[start chain=2 going below,  every join/.style=arrow,]
                \node[blocknode, join, on chain=2, below of=cc, yshift=-0.1cm, 
                      xshift=-2.3cm, text width=3.5cm] (n5u)
                    {Injective Enumeration};
            \end{scope}

            \begin{scope}[start chain=3 going below,  every join/.style=arrow]
                \node[blocknode, text width=3.5cm, join, on chain=3, below of=cc, 
                      yshift=-0.1cm, xshift=2.3cm] (n5d)
                    {Surjective Enumeration};
            \end{scope}
            \draw[->,style=arrow] (n3) -- 
                node[yshift=-0.1cm, above,font=\scriptsize\sffamily,rotate=45] 
                {$\triangleleft \in \{ \geq, = \}$} 
                (n5u);
            \draw[->,style=arrow] (n3) -- 
                node[yshift=-0.1cm, above,font=\scriptsize\sffamily,rotate=315] 
                {$\triangleleft \in \{ \leq, = \}$} 
                (n5d);
            \draw[->,style=arrow] (n5u) -- (n7);
            \draw[->,style=arrow] (n5d) -- (n7);

            \node[textnode, right=0.2cm of n2]{Definition~\ref{defn:symb-TS}};
            \node[textnode, right=0.2cm of n3]{Definitions~\ref{defn:injective} and \ref{defn:surjective}};
            \node[textnode, below=0.0cm of n7, text badly centered]
                {Section~\ref{sec:counting}};
            \node[textnode, below=0.0cm of n5u, xshift=-0.95cm]{
                    Properties~\ref{eqn:inj-te-1} and \ref{eqn:inj-te-3}};
            \node[textnode, below=0.0cm of n5d, xshift=1.45cm]{
                    Properties~\ref{eqn:surj-te-1} and \ref{eqn:surj-te-2}};
        \end{tikzpicture}
    \end{center}
    \caption{Overview of the complete verification methodology.}
    \label{fig:architecture}
\end{figure}
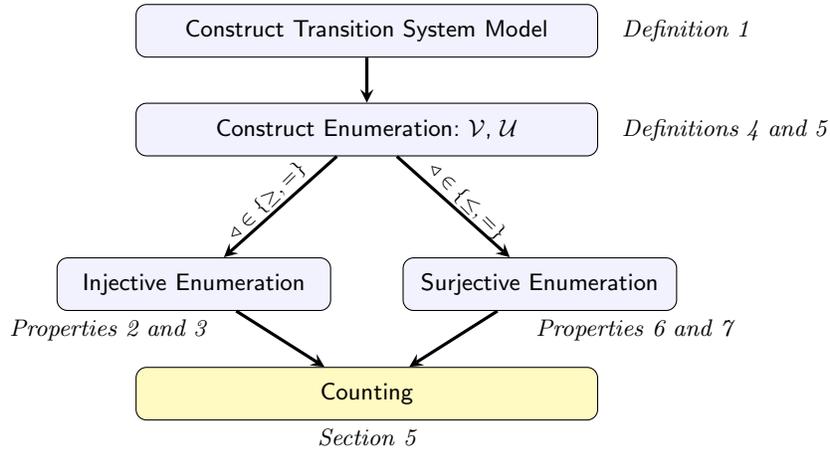
}
{}
To prove a QHP of the form $\qformula{}$, we construct a \emph{trace enumeration predicate} $\valid{\enum{}, Z}{}$ and show an injective/bijective mapping from assignments to $\enum{}$ in $\valid{\enum{}, Z}{}$ and traces of the system.
This allows us to prove $\qformulavalid{\triangleleft}$. 
This part of the proof relies on the notion of a trace enumeration relation (\S~\ref{sec:enum}).
In the next step, we show that ${\countsat{\valid{Y, Z}{}{}}{Y}} \triangleleft N(Z)$ using the inference rules presented in \S~\ref{sec:counting}.
\ifthenelse{\boolean{arxiv}}{
Figure \ref{fig:architecture} shows an overview of the complete methodology and a roadmap for the rest of the paper. 
}{}
We now describe these steps in the context of the motivating example.

\paragraph*{Verification of Soundness for the Z-K Hats Puzzle:}
Property~\ref{eqn:motiv-ex} is illustrated in Figure~\ref{fig:tep-qhp}(b). 
$\trc{0}$ is a trace where the Z-K proof succeeds, while the proof fails for the set of traces $\tracesetcnt{} = \{ \trc{1}, \trc{2}, \dots, \trc{\cntv{}} \}$. The red states show the particular step of the proof in which an incorrect response is given by the prover, and each of these steps as well as their associated prover responses are pairwise different.
The QHP is satisfied if $\sizeof{\tracesetcnt{}} \geq 2^\mathtt{R} - 1$ for every $\trc{0} \in \traces{M}$, where $\mathtt{R} = \valueof{\numrnds{}}{\trcij{0}{0}}$.

The first step in our methodology is to construct a parameterized relation, called a trace enumeration relation, $\tren{\enum{}, \trc{0}, \trc{1}}$. This relates $\trc{0}$ to each trace in the set $\tracesetcnt{}$ and is parameterized by $\enum{}$.
For every value of the parameter $\enum{}$, $\tren{}$ relates a trace in which the proof succeeds ($\trc{0}$) to a trace in which the proof fails ($\trc{1}$).
For every trace $\trc{0}$ in which the proof succeeds, the set $\{ \trc{1} ~|~ \exists \enum{}.~\tren{\enum{}, \trc{0}, \trc{1}} \}$ corresponds to the set of traces with the same challenges and the same number of rounds, but with failed proofs of knowledge. 
Note this is a subset of $\tracesetcnt{}$.

Next, we construct a predicate $\valid{\enum{}, \numrnds{}}{}$ which defines valid assignments to $\enum{}$ for a particular value of $\numrnds{}$.
For a particular $\numrnds{}$, consider the set: $\{ \sigma(\enum{}) ~|~ \sigma \models \valid{\enum{}, \numrnds{}}{} \}$.
Suppose we are able to show that the relation $\tren{}$ is injective in $\enum{}$ and $\trc{0}$ for assignments to $\enum{}$ drawn from this set, then 
we can lower-bound the size of $\tracesetcnt{}$ by the size of this set. 
In other words, we have reduced the problem of trace counting to the problem of counting assignments to $\valid{\enum{}, \numrnds{}}{}$.

Precisely stated, using $\valid{}{}$ and $\tren{}$, we show the following.

\begin{enumerate}
    \item For every trace $\trc{0}$, and every assignment $\enum{}_i$ satisfying $\valid{\enum{}_i, \valueof{\numrnds{}}{\trcij{0}{0}}}{}$, there exists a corresponding trace $\trc{i}$ 
    that satisfies both $\tren{\enum{}_i, \trc{0}, \trc{i}}$ and $\psi(\trc{0}, \trc{i})$. 
    (Note $\valueof{\numrnds{}}{\trcij{0}{0}}$ refers to the valuation of $\numrnds{}$ in the initial state of $\trc{0}$.)
    \item Given two different satisfying assignments to $\valid{}{}$ for a particular value of $\numrnds{}$, say $\enum{}_j$ and $\enum{}_k$, the corresponding traces $\trc{j}$ and $\trc{k}$ 
    are guaranteed to have different prover responses; in other words, the traces satisfy $\delta(\trc{j}, \trc{k})$. 
\end{enumerate}

The above two properties, illustrated in Figure~\ref{fig:tep-qhp}(a), imply there is an injective mapping from satisfying assignments of $\valid{\enum{}, \numrnds{}}{}$ to traces in $\tracesetcnt{}$.
Therefore, the number of traces in $\tracesetcnt{}$ can be lower bounded by the number of satisfying assignments to $\enum{}$ in $\valid{\enum{}, \numrnds{}}{}$, i.e. $\countsat{\valid{\enum{}, \numrnds{}}{}}{\enum{}}$.
We have reduced the difficult problem of counting traces into a {slightly easier} problem of counting satisfying assignments to a $\folt$ formula.

The final step is to bound $\countsat{\valid{\enum{}, \numrnds{}}{}}{\enum{}}$.
For example, one well-known idea from enumerative combinatorics is that if a set $A$ is the union of disjoint sets $B$ and $C$, then $\sizeof{A} = \sizeof{B} + \sizeof{C}$.
Translated to model counting, the above can be written as $\countsat{F(X, Y)} {X}= \countsat{G(X, Y)}{X} + \countsat{H(X, Y)}{X}$ if $F(X, Y) \myiff G(X, Y) \lor H(X, Y)$ is valid and $G(X, Y) \land H(X, Y)$ is \unsat{}.{\footnote{We note there is an implied universal quantifier here. To be precise, we must write $\forall Y.~\countsat{F(X, Y)}{X}= \countsat{G(X, Y)}{X} + \countsat{H(X, Y)}{X}$.}}{}
We present a set of inference rules in Section~\ref{sec:counting} that build on this and related ideas.
These inference rules allow us derive a machine-checked proof of the bound $\countsat{\valid{\enum{}, \numrnds{}}{}}{\enum{}} \geq 2^{\numrnds{}} - 1$, thus completing the proof of Property~\ref{eqn:motiv-ex} for the Z-K hats puzzle.

\section{Overview of Quantitative Hyperproperties}
\label{sec:prelim}

This section introduces a logic for the specification of quantitative hyperproperties over symbolic transition systems.
We present satisfaction semantics for this logic and then discuss its  applications in security verification.

\begin{figure}[htbp]
\setlength{\grammarparsep}{4pt plus 1pt minus 1pt} 
\setlength{\grammarindent}{5em} 
\begin{mdframed}[innerleftmargin=1cm]
\begin{grammar}
  \let\syntleft\relax
  \let\syntright\relax
    <$\psi$> ::= $\forall \trv{}.~\psi$ | 
                 $\countq{\trv{}}{\noteq}{\psi}{\triangleleft}{N(Z)}$ | $\varphi$

    <$\varphi$> ::= 
      $\predsub{\trv{1},\trv{2},\dots,\trv{k}}$ 
      |
    $\lnot \varphi$ | $\varphi \lor \varphi$ |
    $\ltlu{\varphi}{\varphi}$ | $\ltlx{\varphi}$

    <$\triangleleft$> ::= $\leq$ | $=$ | $\geq$
\end{grammar}
\end{mdframed}
\caption{Grammar of Quantitative HyperLTL.}
\label{fig:qhltl}
\ifthenelse{\NOT{\boolean{arxiv}}}{
\vspace{-0.5in} 
}{}
\end{figure}

\subsection{Quantitative Hyperproperties}

Figure~\ref{fig:qhltl} shows the syntax of Quantitative HyperLTL, our extension of HyperLTL~\cite{hyperltl-15} that allows specification of quantitative hyperproperties over symbolic transition systems.
There are two noteworthy differences from the presentation of HyperLTL in~\cite{hyperltl-15}.
The first is the predicate $\predsub{\trv{1}, \trv{2}, \dots, \trv{k}}$.
This refers to a $k$-ary state predicate $\predonly{}$ that is applied to the first element of each trace in the subscript. 
These are analogous to atomic propositions in presentations that use Kripke structures and are defined as $k$-ary state predicates to capture relational properties over traces of the transition system. 
For example, consider the predicate $\predonly(\stat{}_0, \stat{}_1) \doteq (\mathit{input}(\stat{}_0) = \mathit{input}(\stat{}_1))$.
Given this definition, a system $\mname{}$ with exactly two traces $\traces{\mname{}} = \{ \trc{1}, \trc{2} \}$ satisfies the HyperLTL formula $\forall \trv{1}, \trv{2}.~\predsub{\trv{1}, \trv{2}}$ iff $\mathit{input}(\trci{1}{0}) = \mathit{input}(\trci{2}{0})$.
This hyperproperty requires that the input in the initial state of the system be deterministically initialized.

The second difference is the new \emph{counting quantifier}: $\countq{\trv{}}{\noteq}{\psi}{\triangleleft}{N(Z)}$.\footnote{A counting quantifier over Kripke structures was introduced by Finkbeiner et al.~\cite{mcqhyper-2018}. Our definition is slightly different and a detailed comparison is deferred to Section~\ref{sec:related}.} 
$\noteq$ is an unquantified HyperLTL formula over two ``fresh'' trace variables $\trv{j}$ and $\trv{k}$ that encodes when two traces are considered different.\ifthenelse{\boolean{arxiv}}{\footnote{We say that an unquantified HyperLTL formula is over the set of trace variables $V$ iff every variable that appears in the subscript of some predicate $\predonly$ in the formula belongs the set $V$.}}{}
$\psi$ is another (possibly-quantified) HyperLTL formula.
The operator $\triangleleft$ can be $\leq$, $=$, or $\geq$.
$N(Z)$ is an integer-sorted term in $\folt$ over the variables in the set $Z$, $Z \subset X$ where $X$ is the set of state variables of the transition system under consideration.
$Z$ typically refers to the subset of the state variables that define the parameters of the transition system; e.g. $Z = \{ \numrnds{} \}$ for the Z-K proof transition system in Figure~\ref{fig:TS-hat}, the number of blocks in a model of Path ORAM, the size of an array, etc. 
Typically, the variables in the set $Z$ do not change after initialization.
Informally stated, the counting quantifier is satisfied if a maximally large set $\tracesetcnt \subseteq \traceset$ satisfying the following two conditions:
(i) each of the traces in $\tracesetcnt$ are pairwise different as defined by satisfaction of $\noteq$, and
(ii) every trace in this set satisfies the HyperLTL formula $\psi$,
has cardinality $\triangleleft ~\mathit{count}$ where
$\mathit{count}$ is the valuation of $N(Z)$ in the initial state of every trace in $\tracesetcnt{}$.


The remaining operators are standard, so we do not discuss them further and instead provide formal satisfaction semantics.

\subsubsection{Satisfaction Semantics of Quantitative HyperLTL}

\begin{figure}[!t]
    \begin{mdframed}[innerrightmargin=0pt, innerleftmargin=3pt,innertopmargin=0pt,innerbottommargin=3pt]
    \begin{flalign*}
        & \Pi \models_{\traceset{}} 
      \forall \trv{}.~\psi & 
      \text{iff } & 
        \text{for all } \trc{} \in \traceset{}: 
        \Pi[\pi \mapsto \trc{}] \models_{\traceset{}} \psi  & \\
        & \Pi \models_{\traceset{}} 
        \countq{\pi}{\noteq}{\psi}{\triangleleft}{N(Z)}  & 
      \text{iff } & 
        \sizeof{\tracesetcnt} = 0 \myimpl 0 \triangleleft N(Z) \text{ is valid, and } & \\
        & & &
        \sizeof{\tracesetcnt} > 0 \myimpl \forall \trc{} \in \tracesetcnt{}.~\sizeof{\tracesetcnt{}} \triangleleft \valueof{N(Z)}{\trcij{}{0}} \text{, where,} & \\
        & & & \tracesetcnt{} \subseteq \traceset{} \text{ is a maximally large set such that:} \\
        & & &  ~ \forall \trc{j}, \trc{k} \in \tracesetcnt{}.~\\
        & & &  ~~~~~~~~~\trc{j} \neq \trc{k} \myiff \{ \trv{j} \mapsto \trc{j}, \trv{k} \mapsto \trc{k} \} \models\;\noteq{}
             \\
        & & & \text{and, } \forall \trc{} \in \tracesetcnt{}.~ 
            \Pi[\trv{} \mapsto \trc{}] \models_{\traceset{}} \psi  \\
        & \Pi \models_{\traceset{}} 
        \predsub{\trv{1}, \dots,\trv{k}} & 
        \text{iff } & \predonly{}(\Pi(\trv{1})^{0}, \dots, \Pi(\trv{k})^{0}) \text{ is valid } 
      & \\
        & \Pi \models_{\traceset{}} 
      \lnot \psi & 
        \text{iff } & \Pi \not\models_{\traceset{}} \psi 
      & \\
        & \Pi \models_{\traceset{}} 
      \psi \lor \varphi & 
        \text{iff } & \Pi \models_{\traceset{}} \psi \text{ or } \Pi \models_{\traceset{}} \varphi 
      & \\
        & \Pi \models_{\traceset{}} 
        \ltlx{\varphi} & 
        \text{iff } & \Pi^{[1, \infty]} \models_{\traceset{}} \varphi  
      & \\
        & \Pi \models_{\traceset{}} 
        \ltlu{\varphi}{\psi} & 
        \text{iff } & \text{there exists } j \geq 0: \Pi^{[j, \infty]} \models_{\traceset{}} \psi  
      & \\
        & & & \text{ and for all } 0 \leq i < j:  \Pi^{[i, \infty]} \models_{\traceset{}} \varphi
      &  
    \end{flalign*}
    \end{mdframed}
\caption{Satisfaction semantics for Quantitative HyperLTL formulas over symbolic transition systems.}
\label{fig:sem-hyper-ltl}
  \ifthenelse{\NOT{\boolean{arxiv}}}{\vspace{-0.25in}}{}
\end{figure}
The validity judgement of a property $\varphi$ by a set of traces $\traceset{}$ is defined with respect to a trace assignment $\Pi : \varset \to \traceset{}$. 
Here, $\varset$ is the set of trace variables.
We use $\pi, \pi_1, \pi_2$ $, \dots$ to refer to trace variables.\footnote{Note the distinction between trace variables denoted by $\trv{1}, \trv{2}$, etc. and traces which are denoted by $\trc{1}, \trc{2}$, etc.}
The partial function $\Pi$ is a mapping from trace variables to traces.
We use the notation $\Pi[\pi \mapsto \trc{}]$ to refer to a trace assignment that is identical to $\Pi$ except for the trace variable $\trv{}$ which now maps to the trace $\trc{}$.
We write $\Pi \models_{\traceset{}} \psi$ if the set of traces $\traceset{}$ satisfies the property $\psi$ under the trace assignment $\Pi$.
We will drop the subscript $\traceset{}$ from $\models_{\traceset{}}$ if it is clear from the context or irrelevant.
The notation $\Pi^{[i, \infty]}$ is an abbreviation for the new trace assignment obtained by taking the suffix starting from index $i$ of every trace in $\Pi$: $\Pi^{[i, \infty]}(\pi) = \Pi(\pi)^{[i,\infty]}$ for every trace $\pi \in \dom{\Pi}$ where $\dom{\Pi}$ is the domain of $\Pi$. 
We write $\Pi \not\models_{\traceset{}} \psi$ when $\Pi \models_{\traceset{}} \psi$ is not satisfied.
Satisfaction rules for HyperLTL formulas are shown in Figure~\ref{fig:sem-hyper-ltl}.
\begin{definition}[Quantitative HyperLTL Satisfaction]
    We say that the transition system $\mname{}$ satisfies the property $\psi$, denoted by $\mname{} \models \psi$ if the empty trace assignment $\emptyset$ satisfies formula $\psi$ for the set of traces $\traces{M}$, that is $\emptyset \models_{\traces{M}} \psi$.
\end{definition}

\paragraph{Additional Operators:}
The above showed the minimal set of operators required in Quantitative HyperLTL.
The rest of this paper will use the other standard operators such as $\land$ (conjunction), $\myimpl$ (implication), $\ltlf$ (future/eventually) and $\ltlg$ (globally/always) which can be defined in terms of the operators in Figure~\ref{fig:qhltl}.

\paragraph{Well-defined Formulas:} In order for the semantics of Quantified HyperLTL to be meaningful, we need certain semantic restrictions on the structure of QHPs. 
\begin{definition}[Well-defined QHPs]
    An instance of a counting quantifier $\countqinst{}$ is said to be well-defined if:
    \begin{enumerate}
        \item $\lnot\noteq$ is an equivalence relation over the set of all traces $\traces{}$, and
        \item In every set of the traces $\tracesetcnt{}$ captured by the counting quantifier in the semantics shown in Figure~\ref{fig:sem-hyper-ltl}, the term $N(Z)$ has the same valuation for all initial states: $\forall \trc{i}, \trc{j} \in \tracesetcnt{}.~\valueof{N(Z)}{\trcij{i}{0}} = \valueof{N(Z)}{\trcij{j}{0}}$.
    \end{enumerate}

    A Quantified HyperLTL formula is said to be well-defined if every instance of a counting quantifier in the formula is well-defined.
\end{definition}

\begin{example}[Well-defined QHPs]\label{ex:well-defined-qhps}
  The QHPs presented in the rest of this paper are all well-defined, so here we give an example of a QHP that is \emph{not} well-defined. Consider this variant of Property~\ref{eqn:motiv-ex}: $\forall \trv{0}.\countq{\trv{1}}{\true}{\ltlg{(\psi_{\trv{0}, \trv{1}})}}{\geq}{2^\numrnds - 1}$. 
  This is not a well-defined QHP because $\noteq{}$ in the counting quantifier is simply $\true$, and its negation is not an equivalence relation over the set of traces.
\end{example}

Note that condition (1) in the definition above affects $\noteq$ while condition (2) places a restriction on $\varphi$. 
The former condition prevents double-counting of traces, while the latter ensures that the trace count is unambiguous. 

The properties in our experiments require only syntactic checks to verify well-definedness. 
Specifically, $\noteq$ is always of the form $\ltlf{(\predsub{\trv{j}, \trv{k}})}$ where $\predonly{}$ is of the form $\predonly(\stat{}_1, \stat{}_2) \doteq f(\stat{}_1) \neq f(\stat{}_2)$.
The negation of this is obviously an equivalence relation over the set of all traces.
Secondly, our QHPs are of the form $\qformula$ where $\varphi$ enforces equality of the variables in $Z$ between the traces $\trv{0}$ and $\trv{1}$.
These two features guarantee well-definedness.
In the rest of this paper, we only consider well-defined QHPs.

\subsection{Applications of QHPs in Security Specification}
\paragraph{Deniability:}
\label{ex:deniability}
Our first example of a quantitative hyperproperty is deniability.
Suppose $\obs{\stat{}}$ is a term that corresponds to the adversary observable part of the state $\stat{}$, while $\secret{\stat{}}$ corresponds to the secret component of the state $\stat{}$.
    Deniability is satisfied when every trace of adversary observations can be generated by at least $N(Z)$ different secrets.
For this, we define $\delta(\stat{}_1, \stat{}_2) \doteq \secret{\stat{}_1} \neq \secret{\stat{}_2}$ and $\approx^O(\stat{}_1, \stat{}_2) \doteq \obs{\stat{}_1} = \obs{\stat{}_2}$.

\begin{equation*}
    \forall \trv{0}.\countq{\trv{1}}{\ltlf{(\delta_{\trv{j}, \trv{k}})}}{\ltlg{(\approx^O_{\trv{0}, \trv{1}})}}{\geq}{N(Z)}
\end{equation*}

  \begin{figure}[htbp]
    \centering
      \begin{center}
    \begin{tikzpicture}
      \newcommand{\pictext}[1]{\scriptsize{#1}}
      \tikzstyle{state}=[circle,minimum width=0.5cm,inner sep=0,
                         draw,node distance=1.25cm,fill=green!05];
      \foreach \s/\r/\lb/\pos/\col in {
          1/0/1/0/green!05, 
          2/-1.2/2/2/blue!05, 
          3/-2.4/3/k/red!10, 
          n/-4.0/\cntv{}/1/yellow!20}
      {
          \node[state,fill=\col] (s0\s) at (0,\r) 
             {\pictext{$\trcij{\lb}{0}$}};
          \node[left of=s0\s,node distance=0.7cm] (t\s) {\pictext{$\trc{\lb}$}};
          \node[state, fill=\col, right of=s0\s] (s1\s)
             {\pictext{$\trcij{\lb}{1}$}};
          \node[state, fill=\col, right of=s1\s] (s2\s)
             {\pictext{$\trcij{\lb}{2}$}};
          \node[state, fill=\col, right of=s2\s] (s3\s)
             {\pictext{$\trcij{\lb}{3}$}};
          \node[state, fill=\col, right of=s3\s, node distance=2cm] (sk\s)
             {\pictext{$\trcij{\lb}{k}$}};
          
          \node[state,fill=blue!20] at (s\pos\s) {\pictext{$\trcij{\lb}{\pos}$}};

         \draw[->] (s0\s) -- (s1\s);
         \draw[->] (s1\s) -- (s2\s);
         \draw[->] (s2\s) -- (s3\s);
         \draw[->] (s3\s) -- coordinate[midway](sj\s) 
                   node[fill=white] {\pictext{$\dots$}} (sk\s);
         \draw[->] (sk\s) -- ($(sk\s) + (1cm,0)$) 
            node[right] (sdots\s) {\pictext{$\dots$}};

      }

     \foreach \p/\q in {1/2, 2/3, 3/n} {
     \draw[<->] (s0\p) -- 
          node[fill=white]
          {\pictext{\textcolor{red}{$\approx^O$}}} 
          (s0\q);
     \draw[<->] (s1\p) -- 
          node[fill=white]
          {\pictext{\textcolor{red}{$\approx^O$}}} 
          (s1\q);
     \draw[<->] (s2\p) -- 
          node[fill=white]
          {\pictext{\textcolor{red}{$\approx^O$}}} 
          (s2\q);
     \draw[<->] (s3\p) -- 
          node[fill=white]
          {\pictext{\textcolor{red}{$\approx^O$}}} 
          (s3\q);
     \draw[<->] (sk\p) -- 
          node[fill=white]
          {\pictext{\textcolor{red}{$\approx^O$}}} 
          (sk\q);
     }
     \draw[draw=none] (t3) -- 
        node {\pictext{$\vdots$}} 
        ($(tn) + (0, 0.5)$);
    \end{tikzpicture}
  \end{center}
    \caption{Illustrating deniability.}
    \label{fig:deniability}
  \end{figure}
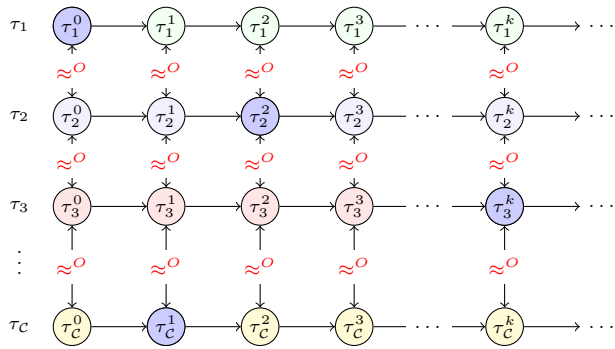

    Figure~\ref{fig:deniability} illustrates deniability.
It shows a set of traces $\tracesetcnt{} := \{ \trcj{1}, \trcj{2}, \ldots, \trcj{\cntv} \}$; the circles represent the states in each trace and the secret values are shown by color of the circle. 
      For these traces, every pair of corresponding states have the same observations: represented by $\approx^{O}$, and every distinct pair of traces differ in the secrets. 
      Deniability is satisfied if $\sizeof{\tracesetcnt{}} \geq N(Z)$.
    Satisfaction implies that every trace of adversary observations has at least $N(Z)$ counterparts with identical observations but different values of $\secret{\stat{}}$. 
    If we can show in a system satisfying deniability that each trace of secrets is equiprobable and $N(Z)$ grows exponentially in some parameters of the system, then we can conclude that the system satisfies computational indistinguishability.
Deniability can capture probabilistic notions of confidentiality, such as confidentiality of Path ORAM.

\paragraph{Soundness:}
\label{ex:soundness}
    While deniability encodes a form of confidentiality, soundness is its dual in the context of integrity. 
    One example of soundness was given in \S~\ref{sec:ex-prop-desc} for the Z-K hats puzzle.
    Soundness is generally applicable to protocols that offer probabilistic integrity guarantees. 
    For instance, many interactive challenge-response protocols which consist of repeated rounds such that if the prover succeeds in all rounds, the verifier can be convinced with high probability that the prover is not cheating.
    This can be viewed as a QHP stating that for every trace in which a dishonest prover tricks a verifier into accepting an invalid proof, there are at least $N(Z)$ other traces with different prover responses in which the cheating is detected. 
    As usual, we require that traces be uniformly sampled from a finite set in order to state soundness as a QHP.

    Soundness is stated as $\forall \trv{0}.\countq{\trv{1}}{\ltlf{(\delta_{\trv{j}, \trv{k}})}}{\ltlg{(\psi_{\trv{0}, \trv{1}})}}{\geq}{N(Z)}$.
    The relation $\delta$ is defined as two states having different prover responses. 
    $\psi$ requires the challenge-response protocol to fail in $\trv{1}$ if it succeeded in $\trv{0}$ and also that the system parameters (the variables in $Z$) be identical between $\trv{0}$ and $\trv{1}$. 

\ifthenelse{\boolean{arxiv}}{
  \begin{figure}[htbp]
    \centering
      \begin{center}
    \begin{tikzpicture}
      \newcommand{\pictext}[1]{\footnotesize{#1}}
      \tikzstyle{state}=[circle,minimum width=0.5cm,inner sep=0.05cm,
                         draw,node distance=1.25cm,fill=green!05];
      \foreach \s/\r/\lb/\pos/\col in {1/0/1/0/green!05, 2/-1.6/2/2/blue!05, 3/-3.2/3/k/red!10, n/-5.1/\cntv{}/1/yellow!20}
      {
          \node[state,fill=\col] (s0\s) at (0,\r) 
             {\pictext{$\trcij{\s}{0}$}};
          \node[left of=s0\s,node distance=0.7cm] (t\s) {\pictext{$\trc{\lb}$}};
          \node[state, fill=\col, right of=s0\s] (s1\s)
             {\pictext{$\trcij{\s}{1}$}};
          \node[state, fill=\col, right of=s1\s] (s2\s)
             {\pictext{$\trcij{\s}{2}$}};
          \node[state, fill=\col, right of=s2\s] (s3\s)
             {\pictext{$\trcij{\s}{3}$}};
          \node[state, fill=\col, right of=s3\s, node distance=2cm] (sk\s)
             {\pictext{$\trcij{\s}{i}$}};
          
         \draw[->] (s0\s) -- (s1\s);
         \draw[->] (s1\s) -- (s2\s);
         \draw[->] (s2\s) -- (s3\s);
         \draw[->] (s3\s) -- coordinate[midway](sj\s) 
                   node[fill=white]  {\pictext{$\dots$}} (sk\s);
         \draw[->] (sk\s) -- ($(sk\s) + (1cm,0)$) 
                    node[right] (sdots\s) {\pictext{$\dots$}};

      }

     \foreach \p/\q in {1/2, 2/3, 3/n} {
     \draw[dashed,<->] (s0\p) -- 
          node[rotate=90,fill=white] {\pictext{\textcolor{red}{$\approx^I$}}} 
          (s0\q);
     \draw[dashed,<->] (s1\p) -- 
          node[rotate=90,fill=white] {\pictext{\textcolor{red}{$\approx^I$}}} 
          (s1\q);
     \draw[dashed,<->] (s2\p) -- 
          node[rotate=90,fill=white] {\pictext{\textcolor{red}{$\approx^I$}}} 
          (s2\q);
     \draw[dashed,<->] (s3\p) -- 
          node[rotate=90,fill=white] {\pictext{\textcolor{red}{$\approx^I$}}} 
          (s3\q);
     \draw[dashed,<->] (sk\p) -- 
          node[rotate=90,fill=white] {\pictext{\textcolor{red}{$\approx^I$}}} 
          (sk\q);
     }
     \draw[draw=none] (t3) -- 
        node {\pictext{$\vdots$}} 
        ($(tn) + (0, 0.5)$);
  \end{tikzpicture}
  \end{center}
      \caption{Illustrating quantitative non-interference.}
    \label{fig:qni}
  \end{figure}
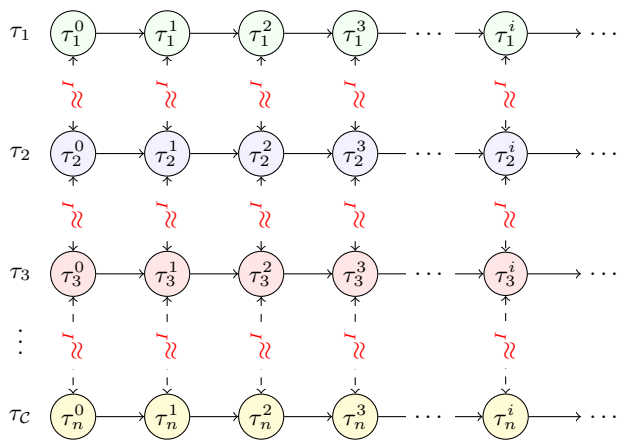
\paragraph{Quantitative Non-interference:}
    In contrast to the above examples, which lower-bounded the number of traces, quantitative non-interference~\cite{yasuoka2014quantitative,  smith2009foundations} \emph{upper-bounds} the amount of information that an attacker can gain from any single trace of attacker-supplied inputs. It is shown in Figure~\ref{fig:qni} and stated as follows:
\begin{equation*}
    \forall \trv{0}.\countq{\trv{1}}{\ltlf{(\delta_{\trv{j}, \trv{k}})}}{\ltlg{(\approx^I_{\trv{0}, \trv{1}})}}{\leq}{N(Z)}
\end{equation*}

In the above, the condition $\approx^I$ encodes the fact that the traces $\trv{0}$ and $\trv{1}$ have the same set of attacker inputs: $\approx^I(\stat{}_1, \stat{}_{2}) \doteq \mathit{inp}(\stat{}_{1}) = \mathit{inp}(\stat{}_{2})$.
    The relation $\delta$ requires the states $\stat{}_1$ and $\stat{}_2$ have different attacker observations: $\delta(\stat{}_1, \stat{}_2) \doteq \mathit{obs}(\stat{}_1) \neq \mathit{obs}(\stat{}_2)$. 
    The property is counting the number of different attacker observable outputs for any given input.
    Assuming that attacker inputs are equiprobable, quantitative non-interference implies that the maximum information an attacker can learn from any single trace is $\lg(N(Z))$ bits.
}
{
}

\subsubsection*{Summarizing QHP Specification:}
These examples demonstrate that QHPs have important applications in security verification.
They capture probabilistic notions of both confidentiality and integrity.
In particular, the following form of QHPs consisting of a single quantifier alternation seems especially relevant for security verification: $\qformula{}$.
Each of the examples of quantitative hyperproperties discussed in the previous subsection -- deniability, soundness, \ifthenelse{\NOT{\boolean{arxiv}}}{as well as others like}{and} quantitative non-interference~\cite{yasuoka2014quantitative,  smith2009foundations} fit in this template.
Therefore, in the rest of this paper, we focus on developing scalable verification techniques for QHPs that follow this template.

\section{Trace Enumerations}
\label{sec:enum}

This section introduces the notion of a trace enumeration, which is a technique that allows us to reduce the problem of counting traces to that of counting satisfiable assignments to a formula in $\folt$.

\subsection{Trace Enumeration Relations}
\label{sec:trel}

We now formalize injective trace enumerations which allows us to lower-bound the number of traces captured by a counting quantifier in a QHP.

\begin{definition}[Injective Trace Enumeration]
    \label{defn:injective}
    Let us consider a transition system  $\machorig{}$ and the relation $\tren{Y, \trc{1}, \trc{2}}$ where $Y$ is a set of variables disjoint from $X$, $\trc{1}$ and $\trc{2}$ are traces of this  transition system.
    Let $\qformulage{}$ be a QHP where $Z \subset X$.
    Suppose $\valid{Y, Z}{}$ is a predicate over the variables in $Y$ and $Z$.
    We say that the pair $\valid{Y, Z}{}$ and $\tren{Y, \trc{1}, \trc{2}}$ form an injective trace enumeration of the system $\mname{}$ for the QHP $\qformulage{}$ iff the following conditions are satisfied:

    \begin{enumerate}
        \item For every trace $\trc{0}$ in $\traces{\mname{}}$ and every satisfying assignment $(\fixy, \fixz)$ for the predicate $\valid{Y, Z}{}$, there exists a trace $\trc{1} \in \traces{\mname{}}$ which is related to the trace $\trc{0}$ as per the relation $\tren{}$ via this same assignment to $Y$. 
            Further, the pair $\trc{0}$ and $\trc{1}$ satisfy the property $\varphi$ and the valuation of the variables in $Z$ in the initial state of $\trc{1}$ is equal to $\fixz$. 
        \begin{align}
            \label{eqn:inj-te-1}
            \forall &\trc{0} \in \traces{\mname{}}, {\fixy}, \fixz. 
                    ~\valid{{\fixy, \fixz}}{} \myimpl & \\
                    & \big(\exists \trc{1} \in \traces{\mname{}}.~ 
                    \tren{{\fixy}, \trc{0}, \trc{1}} 
                    \land \{ \trv{0} \mapsto \trc{0},
                              \trv{1} \mapsto  \trc{1} \}
                           \models \varphi \land
                    \valueof{Z}{\trcij{1}{0}} = \fixz 
                    \big) & \nonumber
        \end{align}
    \item Different assignments to the variables in $Y$ for the formula $\valid{Y, Z}{}$ enumerate different traces in $\tren{Y, \trc{0}, \trc{1}}$, where ``different'' means satisfaction of $\noteq$. 
        \begin{align}
            \label{eqn:inj-te-3}
            \forall & 
                \trc{0}, \trc{1}, \trc{2} \in \traces{\mname{}}, 
                    \fixy{}_1, \fixy{}_2, \fixz{}.~ & \\
            &~~\valid{\fixy{}_1, \fixz}{} \land
               \valid{\fixy{}_2, \fixz}{} \land
                \fixy{}_1 \neq \fixy{}_2  
               & \myimpl \nonumber \\ 
             &~~\tren{\fixy{}_1, \trc{0}, \trc{1}} \land 
                \tren{\fixy{}_2, \trc{0}, \trc{2}} \land
                \valueof{Z}{\trcij{1}{0}} = \fixz{} \land
                \valueof{Z}{\trcij{2}{0}} = \fixz{}
               & \myimpl \nonumber \\ 
            &~~\{ \trv{j} \mapsto \trc{1}, \trv{k} \mapsto \trc{2} \} 
                \models\,\noteq \nonumber 
        \end{align}
    \end{enumerate}
\end{definition}

If $\valid{}{}$ and $\tren{}$ form an injective trace enumeration $\mname{}$ for the property $\qformulage{}$, then for every trace $\trc{0}$, there exist at least as many traces satisfying the counting quantifier as there are satisfying assignments to $Y$ in $\valid{Y, Z}{}$. 
This is made precise in the following lemma.

\begin{restatable}{lemma}{tclb}[Trace Count Lower-Bound]
    \label{lem:tclb}
    If $\valid{Y, Z}{}$ and $\tren{Y,\trc{1},\trc{2}}$ form an injective trace enumeration of the system $\mname{}$ for the QHP $\qformulage{}$ and if $\countsat{\valid{Y, Z}{}}{Y}$ is finite for all assignments to $Z$, then $\mname{} \models \forall \trv{0}.\countq{\trv{1}}{\noteq}{\varphi}{\geq}{\countsat{\valid{Y, Z}{}}{Y}}$.
\end{restatable}

\begin{example}[Injective Trace Enumeration]
    Let $\responseij{1}, \dots, \responseij{\numrnds{}}$ be a trace of correct responses for some particular sequence of challenges for our running example.
Consider the array $\enumi{1}, \enumi{2}, \dots, \enumi{\numrnds}$ where each $\enumi{j} \in \{0, 1\}$. 
    $\enum{}$ is a boolean array of size $\numrnds{}$, and $\enum{}[i]=1$ means that the prover gives an incorrect response to the challenge in round $i$.
    We can define the predicate $\valid{}{}$ as follows.
\begin{align}
    \label{eqn:valid-ex}
    \valid{\enum{},\numrnds{}}{} \doteq\; &
    \big(\exists i.~ 1 \leq i \leq \numrnds \land \enumi{i} \neq 0\big) 
    \land \big(\forall i.~(i < 1 \lor i > \numrnds) \myimpl \enumi{i} = 0\big) 
\end{align}

    The above definition ensures that at least one response is incorrect.
    Notice that for every assignment to $\enum{}$ except the assignment of all zeros, the trace of responses defined by $\forall j.~\responseik{j} = \responseij{j} \oplus \enumi{j}$ (where $\oplus$ is exclusive or) corresponds to a valid trace of the system and satisfies the counting quantifier in Property~\ref{eqn:motiv-ex}. 
Specifically, every such response from the prover is incorrect and will result in the protocol failing.
    We can use the above facts to define the 
    relation $\tren{}$ as follows:
\begin{align}
    \label{eqn:tren-ex}
    \tren{\enum{}, \trc{1}, \trc{2}} \doteq\; 
    & \big(\forall j.~\valueof{\responsei{j}}{\trci{1}{0}} = 
                    \valueof{\responsei{j}}{\trci{2}{0}} \oplus \enumi{j}\big)
    & \land & \\
    & \valueof{\challenge}{\trci{1}{0}} =
      \valueof{\challenge}{\trci{2}{0}} \land
      \valueof{\numrnds{}}{\trci{1}{0}} = 
      \valueof{\numrnds{}}{\trci{2}{0}}
    \land
    (\valueof{S}{\trci{1}{\numrnds{}}} \myimpl 
     \lnot\valueof{S}{\trci{2}{\numrnds{}}})
    & & \nonumber
\end{align}
    The pair $\valid{}{}$ and $\tren{}$ form an injective trace enumeration for the system $M$ (defined in Figure~\ref{fig:TS-hat}) for the Property~\ref{eqn:motiv-ex}.
    This is because different $\enum{}$'s will result in different prover responses for the same challenges.
    By Lemma~\ref{lem:tclb}, we can conclude that Property~\ref{eqn:motiv-ex} is satisfied if $\countsat{\valid{\enum{},\numrnds{}}{}}{\enum{}} \geq 2^{\numrnds{}}-1$
\end{example}

\ifthenelse{\boolean{arxiv}}{
We now define the notion of a surjective trace enumeration, which makes it possible to upper bound the number of traces captured by a counting quantifier.

\begin{definition}[Surjective Trace Enumeration]
    \label{defn:surjective}
    A trace enumeration of the system $\mname{}$ consisting of the pair of predicates $\valid{Y, Z}{}$ and $\tren{Y, \trc{1}, \trc{2}}$ is said to be surjective for the QHP $\qformulale$ on the system $\mname{}$ if the following conditions are satisfied.

  \begin{enumerate}
      \item Every pair of traces of $\mname{}$ which satisfy the formula $\varphi$ can be related via the relation $\tren{}$ for some values $(\fixy, \fixz)$ satisfying $\valid{Y, Z}{}$ such that the valuation of the variables in $Z$ in the initial state of $\trc{1}$ is equal to $\fixz$.
    \begin{align}
        \label{eqn:surj-te-1}
        \forall &\trc{0},\trc{1} \in \traces{\mname{}}.~ & \\
                    & ~~ \{\trv{0} \mapsto \trc{0}, \trv{1} \mapsto \trc{1} \}
                            \models \varphi \myimpl 
                    \big(\exists {\fixy, \fixz}.~
                            \valid{\fixy, \fixz}{} \land 
                            \tren{{\fixy}, \trc{0}, \trc{1}} \land
                            \valueof{Z}{\trcij{1}{0}} = \fixz{}
                            \big) & \nonumber 
    \end{align}
    \item Distinct traces, as defined by satisfaction of $\noteq$, must result in different assignments to $Y$ satisfying $\valid{Y,Z}{}$.
        \begin{align}
        \label{eqn:surj-te-2}
            \forall & 
                \trc{0}, \trc{1}, \trc{2} \in \traces{\mname{}}, \fixy{}_1, \fixy{}_2, \fixz{}.~ & \\
            &~~\{\trv{j} \mapsto \trc{1}, \trv{k} \mapsto \trc{2}\} \models\, \noteq 
            \land \valueof{Z}{\trcij{1}{0}} = \valueof{Z}{\trcij{2}{0}} = \fixz{} 
            & \myimpl \nonumber \\
            &~~\{\trv{0} \mapsto \trc{0}, \trv{1} \mapsto \trc{1}\} \models \varphi 
               \land \{\trv{0} \mapsto \trc{0}, \trv{1} \mapsto \trc{2}\} \models \varphi & \myimpl \nonumber \\
            &~~\tren{\fixy{}_1, \trc{0}, \trc{1}} \land 
               \tren{\fixy{}_2, \trc{0}, \trc{2}} \land
               \valid{\fixy{}_1, \fixz{}}{} \land \valid{\fixy{}_2, \fixz{}}{} & \myimpl \nonumber \\ 
            &~~\fixy{}_1 \neq \fixy{}_2 & \nonumber
        \end{align}
  \end{enumerate}
\end{definition}

Analogous to injective trace enumerations, surjective enumerations can be used to upper-bound the number of traces satisfying the counting quantifier.

\begin{restatable}{lemma}{tcub}[Trace Count Upper-Bound]
\label{lem:tcub}
    If the pair $\valid{Y, Z}{}$ and $\tren{Y,\trc{1},\trc{2}}$ form a surjective trace enumeration of the system $\mname{}$ for the QHP formula\linebreak
     $\qformulale{}$ and if $\countsat{\valid{Y, Z}{}}{Y}$ is finite for every assignment to $Z$, then $\mname{} \models \forall \trv{0}.\countq{\trv{1}}{\noteq}{\varphi}{\leq}{\countsat{\valid{Y, Z}{}}{Y}}$. 
\end{restatable}


\begin{example}[Surjective Trace Enumeration]
    The definitions of $\valid{}{}$ and $\tren{}$ provided in Equations~\ref{eqn:valid-ex} and \ref{eqn:tren-ex} are also surjective trace enumerations, for the transition system shown in Figure~\ref{fig:TS-hat}, with respect to Property~\ref{eqn:motiv-ex}.
    As a result, Lemmas~\ref{lem:tclb} and \ref{lem:tcub} together give us a tight bound of $\countsat{\valid{\enum{}, \numrnds{}}{}}{\enum{}}$, and therefore a tight bound on the number of satisfying traces for the counting quantifier in Property~\ref{eqn:motiv-ex}. 
\end{example}
}
{
  Analogous to injective trace enumerations, it is also possible to define surjective trace enumerations that upper-bound the number of traces captured by a counting quantifier.
  Details of surjective trace enumerations are presented in the extended version of the paper~\cite{qhp-arxiv}.
}

\section{Model Counting}
\label{sec:counting}

As discussed in the previous section, trace enumeration relations can bound the number of satisfying traces in a QHP.
Given a QHP $\qformula$, appropriate trace enumeration predicates $\valid{Y, Z}{}$ and $\tren{}$ can be used 
to derive that $\qformulavalid{\triangleleft}$.
The final step in our verification methodology is to show validity of $\countsat{\valid{Y, Z}{}}{Y} \triangleleft N(Z)$.
To that end, this section discusses our novel technique for model counting.

\subsection{Model Counting via SMT Solving}

Our approach borrows ideas from enumerative combinatorics~\cite{bjorner-2010,wilf-2005,zb-2010} and introduces the inference rules shown in Figure~\ref{fig:mcsem} to reason about model counts for formulas in $\folt$.
Each of the conclusions in the inference rules is a statement involving model counts of $\folt$ formulas, while each of the premises is a formula in $\folt$ that \emph{does not involve model counts} and {can, therefore, be checked using SAT/SMT solvers}.
\ifthenelse{\boolean{arxiv}}{
We describe these inference rules next.

\paragraph{Range:} This rule states that the number of satisfying assignments in the variable $i$ to a formula of the form $a \leq i < b$ is $b - a$ if $b \geq a$ and zero otherwise.
This rule forms one of the ``base cases'' in our derivations.

\paragraph{Positive:} This rule states that the number of satisfying assignments is always greater than or equal to zero.
We will use this in conjunction with other rules which upper bound the number of satisfying solutions to formulas.

\paragraph{ConstLB and ConstUB:} If a formula $f(X)$ has $c$ distinct solutions, we can conclude that the $\countsat{f(X)}{X}$ is lower-bounded by $c$. 
$\mathit{ConstUB}$ is the converse of $\mathit{ConstLB}$. 
It states that if a formula \emph{does not} have $c$ distinct solutions, $f(X)$ definitely has fewer than $c$ satisfying assignments in $X$. 

\paragraph{UB:} If we have two formulas $f(X)$ and $g(X)$ such that $f(X) \myimpl g(X)$, this means that $g(X)$ has at least as many satisfying solutions as $f(X)$.
}
{
Most of the rules are straightforward, we do not describe them due to space constraints.
The three interesting rules --  $\mathit{Injectivity}$, $\mathit{Ind_\leq}$ and $\mathit{Ind_\geq}$ -- are discussed below.

}
\paragraph{Injectivity:}
This rule is based on the following idea from enumerative combinatorics. 
Suppose we have two sets $A$ and $B$. 
We can show that $\sizeof{A} \leq \sizeof{B}$ if there exists an injective function from $A$ to $B$.
Translating this to model counts, the set $A$ in the rule corresponds to satisfying assignments to $f(X)$, $B$ corresponds to satisfying assignments to $g(Y)$ and $\mathscr{F}$ is the injective witness function.

\paragraph{$\mathit{Ind_\geq}$ and $\mathit{Ind_{\leq}}$:} Suppose the formulas $f(X,n)$ and $g(Y,n)$ are parameterized by the integer variable $n$.
If an injective witness function $\mathscr{G}(X, Y, n)$ is able to ``lift'' satisfying assignments of $f(X_n,n)$ and $g(Y_n, n)$ into a satisfying assignment of $f(X_{n+1}, n+1)$, then we can conclude that the number of satisfying assignments to $f(X, n+1)$ are at least as many as the product of the number of satisfying assignments to $f(X, n)$ and $g(Y, n)$.
$\mathit{Ind}_\leq$ is the surjective version of this rule. 
It applies when a satisfying assignment to $f(X_{n+1}, n+1)$ can be ``lowered'' into satisfying assignments to $f(X_n, n)$ and $g(Y_n, n)$ where the values of $X_n$ and $Y_n$ are given by the witness functions $\mathscr{H}_x$ and $\mathscr{H}_y$ respectively.

\begin{figure}[!htb]
  \begin{mdframed}
  \[
  \begin{array}{cc}
  \infer[\mathit{Range}]
  {
      (\countsat{a \leq i < b}{i})~=~\max{(b - a, 0)}
  }
  {
  }
  & 
  ~~~~  
  \infer[\mathit{Positive}]
  {
    \countsat{f(X)}{Y}~\geq~0
  }
  {
  }
  \end{array}
  \]
  \[
  \begin{array}{c}
  \infer[\mathit{ConstLB}]
  {
    \countsat{f(X)}{X}~\geq~c
  }
  {
      \bigwedge_{i=1}^c f(X_i) \land \mathit{distinct}(X_1, \dots, X_c) \text{ is \sat}
  }
  \end{array}
  \]
  \[
  \begin{array}{c}
  \infer[\mathit{ConstUB}]
  {
    \countsat{f(X)}{X}~<~c
  }
  {
      \bigwedge_{i=1}^c f(X_i) \land \mathit{distinct}(X_1, \dots, X_c) \text{ is \unsat}
  }
  \end{array}
  \]
  \[
  \begin{array}{c}
  \infer[\mathit{UB}]
  {
    \countsat{f(X, Y)}{X}~\leq~\countsat{g(X, Y)}{X}
  }
  {
      f(X, Y) \myimpl g(X, Y)
  }
  \end{array}
  \]

  \[
  \begin{array}{c}
  \infer[\mathit{AndUB}]
  {
      \countsat{h(X,Y)}{X \cup Y}~\leq~\countsat{f(X)}{X} \times \countsat{g(Y)}{Y}
  }
  {
    h(X, Y) \myiff f(X) \land g(Y)
  }
  \end{array}
  \]
  \[
  \begin{array}{c}
  \infer[\mathit{Injectivity}]
  {
    \countsat{f(X)}{X}~\leq~\countsat{g(Y)}{Y} 
  }
  {
    \hfil
      f(X)  \myimpl g(\mathscr{F}(X)) 
    \hfil\\
    \big( f(X_1) \land f(X_2) \land X_1 \neq X_2\big) \myimpl \mathscr{F}(X_1) \neq \mathscr{F}(X_2) 
    \vphantom{\text{\Large{B}}}
  }
  \end{array}
  \]
  \[
  \begin{array}{c}
  \infer[\mathit{Disjoint}]
  {
      \countsat{h(X, Y)}{X \cup Y}~=~\countsat{f(X)}{X} \times \countsat{g(Y)}{Y}
  }
  {
    h(X, Y) \myiff f(X) \land g(Y)      & X \cap Y = \emptyset
  }
  \end{array}
  \]
  \[
  \begin{array}{c}
  \infer[\mathit{Or}]
  {
    \countsat{f(X, Y)}{X}~=~
        \countsat{g(X, Y)}{X} + \countsat{h(X, Y)}{X} - \countsat{\big( g(X, Y) \land h(X, Y)\big)}{X}
  }
  {
    f(X, Y) \myiff g(X, Y) \lor h(X, Y)
  }
  \end{array}
  \]
  \[
  \begin{array}{c}
  \infer[\mathit{Ind_\geq}]
  {
    \countsat{f(X,n+1)}{X}~\geq~\countsat{f(X, n)}{X} \times \countsat{g(Y, n)}{Y}
  }
  {
    \hfil
    \big(f(X, n) \land g(Y, n)\big) \myimpl f(\mathscr{G}(X, Y, n), n+1) 
    \hfil \\ 
      (X_1 \neq X_2 \lor Y_1 \neq Y_2) 
      \myimpl 
      \mathscr{G}(X_1, Y_1, n) \neq \mathscr{G}(X_2, Y_2, n)
      \vphantom{\text{{\Large b}}}
  }
  \end{array}
  \]
  \[
  \begin{array}{c}
  \infer[\mathit{Ind_\leq}]
  {
    \countsat{f(X,n+1)}{X}~\leq~\countsat{f(X, n)}{X} \times \countsat{g(Y, n)}{Y}
  }
  {
    \hfil
    f(X, n+1) \myimpl 
      \big(f(\mathscr{H}_x(X, n+1), n) \land g(\mathscr{H}_y(X, n+1), n) \big)
    \hfil \\ 
      X_1 \neq X_2
    \myimpl 
      \big(
        \mathscr{H}_x(X_1, n) \neq \mathscr{H}_x(X_2, n) 
        \lor
        \mathscr{H}_y(Y_1, n) \neq \mathscr{H}_y(Y_2, n) 
      \big)
    \vphantom{\text{{\Large b}}}
  }
  \end{array}
  \]
  \end{mdframed}
  \caption{Model counting proof rules. Unless otherwise specified, premises are satisfied when the formula is valid. Conclusions have an implicit universal quantifier.}
  \label{fig:mcsem}
\end{figure}
\ifthenelse{\NOT{\boolean{arxiv}}}{\vspace{-0.35in}}{}

\subsection{Model Counting in the Motivating Example}
The definition of the predicate $\valid{}{}$ in the motivating example is shown below. 
\begin{align*}
    \valid{\enum{},\numrnds{}}{} \doteq~ & \big(\exists i.~ 1 \leq i \leq \numrnds \land \enumi{i} \neq 0\big) \land
    \big(\forall i.~((i < 1 \lor i > \numrnds) \myimpl \enumi{i} = 0) \big) \nonumber 
\end{align*}

Our task is to show $\countsat{\valid{\enum{}, \numrnds{}}{}}{\enum{}} = 2^\numrnds{} - 1$.
Recall that $\enum{}$ is an array of binary values (i.e. the integers $0$ and $1$) and consider the following predicates:
$\verf{} \doteq~ \big(\forall i.~(i < 1 \lor i > \numrnds) \myimpl \enumi{i} = 0\big)$, 
$\valid{\enum{},\numrnds{}}{1}  \doteq~ \big(\forall i.~\enumi{i} = 0\big)$ and
$\mathcal{W}(i) \doteq~ 0 \leq i < 2$.
Using these definitions, the proof is as follows.
\begin{enumerate}
    \item ($\mathit{ConstUB}$, $\mathit{Positive}$) $\countsat{\verf{} \land \vero}{\enum{}} = 1$.
    \item ($\mathit{Or}$) $\cverf = \cver{} + \cvero{}$.
    \item ($\mathit{ConstLB}$, $\mathit{ConstUB}$) $\cvero{} = 1$.
    \item ($\mathit{ConstLB}$, $\mathit{ConstUB}$) $\cverfo{} = 2$.
    \item ($\mathit{Ind}_\leq$): $\cverf{} \leq \countsat{\mathcal{W}(i)}{i}\times\cverfp$.
    \item ($\mathit{Ind}_\geq$): $\cverf{} \geq \countsat{\mathcal{W}(i)}{i}\times\cverfp$.
    \item ($\mathit{Range}$): $\countsat{\mathcal{W}(i)}{i} = 2$.
    \item (4 -- 7) imply that $\cverf{} = 2\times\cverfp{}$, $\cverfo{} = 2$, this means $\cverf{} = 2^\numrnds{}$.
    \item (2, 3, 8) imply that $\cver{} = 2^\numrnds{} - 1$.
\end{enumerate}

In step 5, the witness function is $\mathscr{G}(\enum{}, \numrnds{}, i) \doteq \enum{}[\numrnds{} + 1 \mapsto i]$, while in step  6, they are $\mathscr{H}_{\langle \enum{}, \numrnds{} \rangle}(\enum{}, \numrnds{} + 1) \doteq \langle \enum{}[\numrnds{} + 1 \mapsto 0], \numrnds{} \rangle$ and $\mathscr{H}_i(\enum{}, \numrnds{}+1) \doteq (\enum{}[\numrnds+1])$.\footnote{The notation $\mathit{arr}[i \mapsto v]$ denotes an array that is identical to $\mathit{arr}$ except for index $i$ which contains $v$.}
Note steps 8 and 9 are automatically discharged by the SMT solver.

\section{Experimental Results and Discussion}
\label{sec:eval}

In this section, we present an experimental evaluation of the use of trace enumerations for the verification of quantitative hyperproperties.

\subsection{Methodology}
  We studied five systems with varying complexity and QHPs. 
  These were modeled in the \uclid{} modeling and verification framework~\cite{uclid5-memocode18, uclid5-www}, which uses the Z3 SMT solver (v4.8.6)~\cite{z3-tacas-08} to discharge the proof obligations. 
  The experiments were run on an Intel i7-4770 CPU @ 3.40GHz with 8 cores and 32 GB RAM. 

  The verification conditions are currently manually generated from the models, but automation of this is straightforward and ongoing.
  The $k$-trace properties were proven using self-composition~\cite{BartheCsfw04, BartheFM11} and induction.
  A number of strengthening invariants had to be specified manually for the inductive proofs.
  Many of the invariants are relational \emph{and} quantified and, therefore, difficult to infer algorithmically.
  We note that recent work has made progress toward automated inference of quantified invariants~\cite{quic3-18,qsygus-19}. 

\ifthenelse{\boolean{arxiv}}{
\subsubsection{Implementation Issues:}
In addition to the techniques in Sections~\ref{sec:enum} and \ref{sec:counting}, we must address two additional practical challenges.
    The first challenge is the definition of $\tren{Y, \trc{1}, \trc{2}}$ as some general relation over traces. 
    This poses difficulties in our proofs which rely on induction.
    Hence, we impose a syntactic restriction on the relation $\tren{}$ by constraining it to be a relational hyperinvariant of the transition system.
In other words, $\tren{}$ has the  following form:
$\tren{Y, \trc{1}, \trc{2}} \doteq \forall i.~\trel{Y, \valueof{X}{\trci{1}{i}}, \valueof{X}{\trci{2}{i}}}$.
This allows using induction and self-composition~\cite{BartheCsfw04, BartheFM11, TerauchiSas05} to verify that a relation is indeed a trace enumeration predicate.

    The second challenge is the quantifier alternation. Definitions~\ref{defn:injective} and \ref{defn:surjective} involve quantifier alternation over trace variables and verification of such hyperproperties is challenging.
    We address this problem by manually specifying Skolem witness functions for the existential quantifiers~\cite{skolem-67}.
Note that utilizing a Skolem function in this context is not straightforward because we need to construct a witness function for an infinitely long trace.
We side-step the issue by constructing a witness function for the initial state of the corresponding trace and imposing an additional condition requiring that if the initial states of two traces are related via the relation $\trel{}$, then all subsequent states must also be related.
}
{}

\subsection{Overview of Results}
Due to limited space, we only provide a brief description of our benchmarks for evaluation and refer the interested reader to
\ifthenelse{\boolean{arxiv}}{
Appendix \ref{appendix:microbenchmarks} for a more detailed discussion.
}
{
our full paper \cite{qhp-arxiv} for a more detailed discussion.
}
We have also made the models and associated proof scripts available at \cite{experiments-www}. 
A brief overview of the case studies follows. 
\begin{table}[htbp]
    \caption{Verification Results of Models.}
    \begin{center}
       \begin{tabular}{llp{1cm}p{1cm}p{1cm}p{1cm}p{2cm}}
       \toprule
          Benchmark & Hyperproperty & Model LoC & Proof LoC & Num. Annot. & Verif. Time \\ 
       \midrule
          Electronic Purse \cite{backes-09} ~         & Deniability                   & 46  & 93  & 9   & 3.92s \\
          Password checker \cite{mcqhyper-2018} ~     & Quantitative non-interference & 59  & 100 & 10  & 4.69s \\
          F-Y Array Shuffle                           & Quantitative information flow & 86  & 195 & 96  & 7.38s \\
          ZK Hats (Sec. \ref{sec:motivating-example}) & Soundness                     & 91  & 191 & 36  & 6.34s \\
          Path ORAM \cite{path-oram-2013} ~           & Deniability                   & 587 & 209 & 142 & 9.74s \\
       \bottomrule  
      \end{tabular}
    \end{center} 
    \label{tab:results-model}
\end{table}

\begin{enumerate}
  \item {\bf Electronic Purse.} We model an electronic purse, with a secret initial balance, proposed by Backes et al.~\cite{backes-09}. A fixed amount is debited from the purse until the balance is insufficient for the next transaction. We prove a deniability property: there is a sufficient number of traces with identical attacker observations but different initial balances.
  \item {\bf Password Checker.} We model the password checker from Finkbeiner et al.~\cite{mcqhyper-2018}, but we allow passwords of unbounded length $n$. We prove quantitative non-interference: information leakage to an attacker is $\leq n$ bits.
  \item {\bf Array Shuffle.} We implement a variant of the Fisher-Yates shuffle. 
  We chose this because producing random permutations of an array is an important component of certain cryptographic protocols (e.g., Ring ORAM~\cite{roram-15}). We prove a quantitative information flow property stating that all possible permutations are indeed generated by the shuffling algorithm.
  \item {\bf ZK Hats.} We prove soundness of the zero-knowledge protocol in Section~\ref{sec:motiv}. 
  \item {\bf Path ORAM.} Discussed in Section \ref{sec:pathoram}.
\end{enumerate}
The properties we prove on these models and the results of our evaluation are presented in Table \ref{tab:results-model} which
shows the size of each model, the number of lines of proof code (this is the code for self-composition, property specification, etc.), the number of verification annotations (invariants and procedure pre-/post-conditions) and the verification time for each example.
  Once the auxiliary strengthening invariants are specified, the verification completes within a few seconds. 
  This suggests that the methodology can scale to larger models, and even implementations.
  The main challenge in the application of the methodology is the construction of the trace enumeration relations, associated witness functions, and the specification of strengthening invariants.
  Each of these requires application-specific insight.
  Since most of our enumerations and invariants are quantified, some of the proofs also required tweaking the SMT solver's configuration options (e.g. turning off model-based quantifier instantiation in Z3).
\subsection{Deniability of Path ORAM}
    \label{sec:pathoram}
    In this section,  we discuss our main case study: the application of trace enumerations for verifying deniability of server access patterns in Path ORAM \cite{path-oram-2013}, a practical variant of Oblivious RAM (ORAM) \cite{oram-jacm-96}.
    ORAMs refer to a class of algorithms that allow a client with a small amount of storage to store/load a large amount of data on an untrusted server while concealing the client access pattern from the server. 
    Path ORAM stores encrypted data on the server in an augmented binary tree format. Each node stores $Z$ data blocks, referred to as \emph{buckets} of size $Z$. Additionally, the client has a small amount of local storage called the \emph{stash}. 
    The client maintains a secret mapping called the \emph{position map} to keep track of the path where a data block is stored on the server.
    Each entry in the position map maps a client address to a leaf on the server.
    Path ORAM maintains the invariant that every block is stored somewhere along the path from the root to the leaf node that the block is mapped to by the position map.
    \ifthenelse{\boolean{arxiv}}{
    The position map is initialized randomly, and an entry of the position map is updated after every access to that location.

    \begin{algorithm}
      \DontPrintSemicolon
      \SetAlgorithmName{Algorithm}{}{}
      \SetKwInOut{Secret}{Secret Input}
      \SetKwInOut{Input} {Public Input}
      \SetKwInOut{Output}{Public Output }
      
      \Secret{OP, $\request{}$, data$^*$}
      \Output{$\leaf{}$ \tcc*[f]{Accessed Leaf}}

      \BlankLine
        $\leaf{} \gets \pmap{}[\request{}]$ \;
        $\pmap{}[\request{}] \gets$ {\bf UniformRandom}(1, $\numblks{}$) \tcc*[f]{remap}\;
        stash $\gets$ stash $\cup$ {\bf ReadPath}($\leaf{}$) \;
        \If{$OP = \code{WRITE}$}{
          {\bf UpdateData}(stash, $\request{}$, data$^*$) \;
        }
        {\bf WritePath}($\leaf{}$, stash) \;

      \caption{\textsc{Oram Access Protocol}}
      \label{alg:pathoram}
    \end{algorithm}

    The Path ORAM access function shown in Algorithm~\ref{alg:pathoram} can be logically divided into 5 steps: (1) getting the mapped leaf corresponding to the requested data block, (2) updating the position map entry for the accessed data block, (3) reading the path on the server corresponding to the leaf, (4) updating the data in case of a write operation, and (5) writing the data blocks back from stash to the path read from ORAM.

    Our model of Path ORAM is a transition system where each step corresponds to a single execution of the access function for an arbitrary operation and address.
    The model has an unbounded size tree and stash.
    Each bucket contains four nodes -- this is the recommended configuration for Path ORAM~\cite{path-oram-2013}.
    The model makes uses of various uninterpreted functions for abstraction, e.g. modeling the path from the root of the tree to a leaf.
    }
    {}
  
    \subsubsection{Deniability of Server Access Patterns in Path ORAM:}
    \ifthenelse{\boolean{arxiv}}{
    Intuitively, the security of path ORAM requires that the adversary (untrusted server) learns nothing about the access pattern of the client by observing server accesses.
    In Path ORAM, the position map is initialized randomly and every subsequent update to the position is also done by sampling from a uniform random distribution. Hence every secret position map is equally likely to be chosen. 
    }{}
    We formulate security of access patterns in Path ORAM as a deniability property stating that for every infinitely-long trace of server accesses, there are $(\numblks - 1)!$ traces of client accesses with identical server observations but different client requests.
    \begin{equation}
      \forall \trv{0}.~\#\trv{1}: \ltlf{(\delta{}_{\trv{j}, \trv{k}})}.~
              \ltlg{(\psi_{\trv{0}, \trv{1}})}~\geq~(\numblks - 1)!
      \label{eqn:pathoram-property}
    \end{equation}
    The binary predicate $\delta$ imposes the requirement that the client's request are different in each of the traces captured by the counting quantifier, 
    and the condition in $\psi$ states that all the traces captured by the counting quantifier have the same observable access pattern as $\trv{0}$.

    \subsubsection{Verification of Deniability in Path ORAM:}
    To verify the QHP stated in Equation \ref{eqn:pathoram-property}, for every trace of server accesses we need to generate $(\numblks-1)!$ traces of client requests that produce the same server access.
    
    Suppose we have Path ORAM (a) that is initialized with some position map.
    Now consider the Path ORAM (b) with the same number of blocks, but with an initial position map that is a derangement of the position map of (a).\footnote{A {\bf derangement} of a set is a permutation of the elements of the set such that no element appears in its original position.}
    The key insight is that ORAM (b) can simulate an identical server access pattern as ORAM (a) by appropriately choosing a different client request that maps to the same leaf that is being accessed by (a) and then updating the position map identically as (a).
    This is shown in Figure~\ref{fig:pathoram-demonstration}, which shows two Path ORAMs that produce identical server access patterns but service different client requests.

    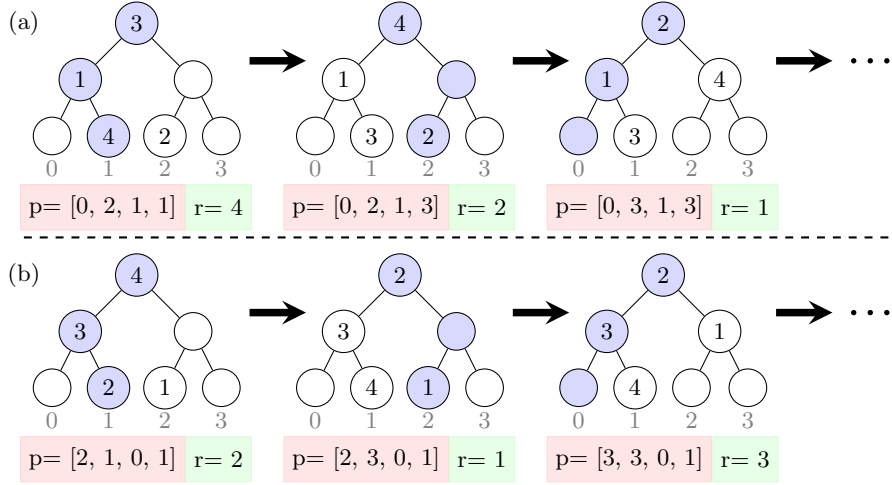
\begin{figure}[!t]
      \centering
      \tikzstyle{filled}=[fill=blue!15]
\tikzstyle{acc}=[fill=yellow!15]
\tikzstyle{arrow}=[line width=1mm,->,>=stealth]
\tikzstyle{textNode}=[rectangle, minimum height=0.6cm]

\newcommand\botFigY{-3.35}

\begin{tikzpicture}[
    every node/.style = {draw, circle, minimum size = 5mm},
    level 1/.style = {sibling distance=1.5cm},
    level 2/.style = {sibling distance=0.75cm}, 
    level distance = .75cm
  ]

  \foreach \n/\x/\y/\S/\sf/\A/\af/\D/\df/\B/\bf/\C/\cf/\E/\ef/\F/\ff/\pz/\po/\pt/\pth/\req in {
          1/0/0/3/filled/1/filled/ // //4/filled/2////0/2/1/1/4,
          2/3.5/0/4/filled/1///filled///3//2/filled///0/2/1/3/2,
          3/7/0/2/filled/1/filled/4///filled/3///// /0/3/1/3/1,
          4/0/\botFigY/4/filled/3/filled/ // //2/filled/1// / /2/1/0/1/2,
          5/3.5/\botFigY/2/filled/3///filled///4//1/filled// /2/3/0/1/1,
          6/7/\botFigY/2/filled/3/filled/1///filled/4///// /3/3/0/1/3
      }
      {
        \node (Start) at (\x,\y) [\sf] {\footnotesize \S}
            child {   node (A) [\af] {\footnotesize \A}
                child { node (B) [\bf] {\footnotesize \B}}
                child { node (C) [\cf] {\footnotesize \C}}
            }
            child {   node (D) [\df] {\footnotesize \D}
                child { node (E) [\ef] {\footnotesize \E}}
                child { node (F) [] {\footnotesize \F}}
            };

        \begin{scope}[nodes = {draw = none}]
            \path (Start) -- (A);
            \path (A)     -- (B);
            \path (A)     -- (C);
            \path (Start) -- (D);
            \path (D)     -- (E);
            \path (D)     -- (F);
            \begin{scope}[nodes = {below = 4pt}]
                \node [text=black!50] at (B) {\footnotesize $0$};
                \node [text=black!50] at (C) {\footnotesize $1$};
                \node [text=black!50] at (E) {\footnotesize $2$};
                \node [text=black!50] at (F) {\footnotesize $3$};
            \end{scope}
        \end{scope}

        \node (pmap\n) [textNode,draw=red!15,fill=red!10, xshift=-.45cm]  at (\x,\y-2.45) {\footnotesize p= [\pz, \po, \pt, \pth]};
        \node (req\n) [textNode,draw=green!15,fill=green!10, right=0cm of pmap\n] {\footnotesize r= \req};
      }

      \node [draw=none] (lab1) at (-1.5,0) {(a)};
      \draw [arrow] (1.5,-0.5) -- (2.25,-0.5);
      \draw [arrow] (5,-0.5) -- (5.75,-0.5);
      \draw [arrow] (8.5,-0.5) -- (9.25,-0.5);
      \node [draw=none] (etc1) at (9.75,-0.5) {\Huge ...};

      \node [draw=none] (lab1) at (-1.5,\botFigY) {(b)};
      \draw [arrow] (1.5,\botFigY-0.5) -- (2.25,\botFigY-0.5);
      \draw [arrow] (5,\botFigY-0.5) -- (5.75,\botFigY-0.5);
      \draw [arrow] (8.5,\botFigY-0.5) -- (9.25,\botFigY-0.5);
      \node [draw=none] (etc2) at (9.75,\botFigY-0.5) {\Huge ...};

      \draw [thick, dashed] (-1.5,\botFigY+0.5) -- (10,\botFigY+0.5);
 
\end{tikzpicture}
      \caption{Path ORAM systems satisfying the counting quantifier of Equation \ref{eqn:pathoram-property}. 
      }\label{fig:pathoram-demonstration}
    \end{figure}

    

    \ifthenelse{\boolean{arxiv}}{
    The above insight leads to the following trace enumeration.
    $\valid{\perm, \permhat, \numblks }{}$ captures the notion of a valid derangement:
    \begin{align}
      & \valid{\perm, \permhat, \numblks}{} \:\doteq\nonumber \\
      & ~~~ \forall i .\: 1 \leq i \leq \numblks \,.\,\big(\perm[i] \neq i \land \permhat[i] \neq i \big)&\land \nonumber\\
        & ~~~ \forall i .\: 1 \leq i \leq \numblks \implies 1 \leq \perm[i] \leq \numblks
        \land 1 \leq \permhat[i] \leq \numblks &\land \nonumber \\
      & ~~~ \forall i,j .\: (1 \leq i \leq \numblks \land 1 \leq j \leq \numblks \land i \neq j) \implies& \nonumber \\ 
      & ~~~~~~~~~~~~~~~~~~~~~~~~~~~~~~~~~~~   \big((\perm[i] \neq \perm[j]) \land (\permhat[i] \neq \permhat[j]) \big)&\land \nonumber \\
      & ~~~ \forall i .\: i < 0 \lor i > \numblks \implies \big( (\perm[i] = 0) \land (\permhat[i] = 0) \big) & \land \nonumber \\
      & ~~~ \forall i. \: 1 \leq i \leq \numblks \,. \Big(\exists j. \: 1 \leq j \leq \numblks .\: \big(\permhat[i] = j \implies \perm[j] = i\big)\Big)
    \end{align}
    Every satisfying assignment to the above is a derangement $\perm{}$ and its inverse $\permhat{}$. 
    These can be used to permute the client's accesses in ORAM (a) so that the server accesses made by ORAM (b) are identical to those of ORAM (a).

    Next we define the relation $\tren{}$.
    \begin{align}
      \tren{\{\perm, \permhat\}, \trc{1}, \trc{2}} \doteq\; 
      \label{eqn:tren-poram}
      & \valueof{\numblks}{\trci{1}{0}} = \valueof{\numblks}{\trci{2}{0}}  \land
        \valueof{\stashsz}{\trci{1}{0}} = \valueof{\stashsz}{\trci{2}{0}} 
      & \land & \nonumber \\  
      & \big(\forall i.~\forall j.~\valueof{\pmap{}[j]}{\trci{2}{i}} = 
                      \valueof{\pmap{}[\perm[j]]}{\trci{1}{i}} \big)
      & \land & \nonumber \\
      & \big(\forall i.~\valueof{\request{}}{\trci{2}{i}} = 
      ~\valueof{\permhat[\request{}]}{\trci{1}{i}} \big)
      & \land & \nonumber \\
      & \big(\forall i.~\valueof{\leaf{}}{\trci{1}{i}} =
        ~\valueof{\leaf{}}{\trci{2}{i}} \big) 
      & \land & \nonumber \\
      & \big(\forall i.~\valueof{\remap}{\trci{1}{i}} =
        ~\valueof{\remap}{\trci{2}{i}} \big)
      & & 
    \end{align}

$\tren{}$ relates two traces which are such that the position map of one is a permutation of the other, their client requests are appropriately permuted, and have identical access patterns and position map updates.
    }
    {
      The above insight leads to a trace enumeration where two traces are related via $\tren{}$ if their position maps are derangements of each other, the client accesses are permuted as per the derangement while all other parameters of the ORAM are identical. 
      We use this to prove Property~\ref{eqn:pathoram-property}.
    }

    \ifthenelse{\boolean{arxiv}}{
    \subsubsection{Discussion:}
    It is important to note that deniability of the access pattern is just one aspect of the security of Path ORAM.
    Path ORAM also requires that the data stored on the server be encrypted and authenticated using randomized authenticated encryption algorithm.
    Verifying these aspects of Path ORAM is unrelated to quantitative hyperproperties, so we do not incorporate them in our model.
    }
    {
    }


\section{Related Work}
\label{sec:related}

\noindent \textbf{Hyperproperties}:
Research into secure information flow started with the seminal work of Denning and Denning~\cite{Denning77}, Goguen and Meseguer~\cite{goguen-sp82} and Rushby~\cite{rushby-82}.
The self-composition construction for the verification of secure information flow was introduced by Barthe et al.~\cite{BartheCsfw04}.
\ifthenelse{\boolean{arxiv}}{
Terauchi and Aiken identified the class of $k$-safety properties~\cite{TerauchiSas05}, which is an important subset of the class of hyperproperties.
}{}
Clarkson and Schneider~\cite{ClarksonS10} introduced the class of specifications called hyperproperties\ifthenelse{\boolean{arxiv}}{and showed that both noninterference and observational determinism~\cite{obsdet-03,roscoe-sp-95}, as well many other security specifications were instances of hyperproperties.}{.}
Clarkson and colleagues also introduced HyperLTL and HyperCTL$^*$~\cite{clarkson-14}, which are temporal logics for specifying hyperproperties, while verification algorithms for these were introduced by Finkbeiner and colleagues in~\cite{hyperltl-15}.
Cartesian Hoare Logic~\cite{CHL16} was introduced by Sousa and Dillig and enables the specification and verification of hyperproperties over programs as opposed to transition systems.
A number of subsequent efforts have studied hyperproperties in the context of program verification~\cite{lazysc-cav18,shemer-cav-19,farzan-cav-19,decomp-17}.

\noindent \textbf{Quantitative Information Flow}:
Quantitative hyperproperties build on the rich literature of quantitative information flow (QIF)~\cite{smith2009foundations, alvim2012quantitative, clarkson2005belief, clark2005quantitative, gray1992toward}. 
The QIF problem is to quantify (or bound) the number of bits of secret information that is attacker-observable.
Certain notions of QIF can be expressed as QHPs. 
It is important to note QHPs can express security specifications (e.g., soundness) that are not QIF.
Yasuoka and Terauchi studied QIF from a theoretical perspective and showed that it could be expressed as hypersafety and hyperliveness~\cite{yasuoka2014quantitative}. 
\ifthenelse{\boolean{arxiv}}{
They have the first construction to show that QIF with a constant bound of $b$ bits can be expressed as a $k$-safety property where $k=2^b + 1$.
In principle, this means that QIF can be expressed as HyperLTL/HyperCTL~\cite{clarkson-14} formulas and verified using the self composition-based algorithm in~\cite{hyperltl-15} assuming the bound is static.
}{}
Approaches based on QIF measures such as min-entropy \cite{smith2009foundations}, Shannon entropy~\cite{clark2007static} etc. have also been applied in the context of static analysis~\cite{kopf2012automatic}.
\ifthenelse{\boolean{arxiv}}{
QIF has been studied for specific applications; e.g. cache-based side-channel attacks~\cite{doychev2015cacheaudit,kopf2012automatic} and web applications~\cite{chen2010side, zhang2010sidebuster, phan2014abstract}.
These stand in contrast to our approach, which permits verification of a large and generic class of QHPs. 
}{}

\noindent \textbf{Quantitative Hyperproperties}:
Quantitative Cartesian Hoare Logic (QCHL) enables verification of certain quantitative properties of programs~\cite{QCHL17}.
QHPs are more expressive than QCHL, the latter counts events within a trace (e.g. memory accesses), while QHPs count the number of traces satisfying certain conditions.

The most closely related work to ours is of Finkbeiner et al.~\cite{mcqhyper-2018} who introduced Quantitative HyperLTL over Kripke structures.
They also introduced a verification algorithm for this logic that is based on maximum model counting.
However, their algorithm does not scale to reasonable-sized systems, and experiments from their paper show that the approach times out when checking an 8-bit leak in a password checker (using 8-bit passwords). 
We differ from their work in three important ways. 
First, our properties are defined over symbolic transition systems rather than Kripke structures. 
This allows modeling and verification of QHPs over infinite-state systems.
Second, our bounds are symbolic, which enables us to express bounds as functions of transition system parameters. 
Finally, our definition of Quantitative HyperLTL is also more expressive.
It is not possible to convert our QHPs into (non-quantitative) HyperLTL formulas with $k$-traces for any fixed value of $k$.

\noindent \textbf{Verification of ORAMs}: In concurrent work with ours, Barthe et al.~\cite{barthe-19} and Darais et al.~\cite{hicks-20} have introduced specialized mechanisms to prove security of ORAMs. 
Barthe et al.~\cite{barthe-19} introduced a probabilistic separation logic (PSL) that (among other things) can be used to reason about the security of ORAMs.
Unlike QHPs, PSL does not permit quantitative reasoning about probabilities of events and also does not (yet) support machine-checked reasoning.
Darais et al.~\cite{hicks-20} introduce a type system that enforces obliviousness; they use this type system to implement a tree-based ORAM.
Note that QHPs can express specifications other than obliviousness, and obliviousness need not necessarily be a QHP.


\section{Conclusion}
\label{sec:concl}

Quantitative hyperproperties are a powerful class of specifications that stipulate the existence of a certain number of traces satisfying certain constraints.
Many important security guarantees, especially those involving probabilistic guarantees of security, can be expressed as quantitative hyperproperties.
Unfortunately, verification of quantitative hyperproperties is a challenging problem because these specifications require simultaneous reasoning about a large number of traces of a system. 
In this paper, we introduced a specification language, satisfaction semantics, and a verification methodology for quantitative hyperproperties. 
Our verification methodology is based on reducing the problem of counting traces into that of counting the number of assignments that satisfy a first-order logic formula.
Our methodology enables security verification of many interesting security protocols that were previously out of reach, including confidentiality of access pattern accesses in Path ORAM.

\section*{Acknowledgements}
We sincerely thank the anonymous reviewers for their insightful comments which helped improve this paper.
This work was supported in part by the Semiconductor Research Corporation under Task 2854 and the Science and Engineering Research Board of India, a unit of the Department of Science and Technology, Government of India.


\bibliographystyle{plain}
\bibliography{references}

\ifthenelse{\boolean{arxiv}}
{
\begin{subappendices}
    \renewcommand{\thesection}{\Alph{section}}%
    \section{Proofs}
\label{app:proofs}

\begin{lemma}[Equivalence Class Characterization]
    \label{lem:eqclass}

    Let $\machorig$ be a transition system and $\traces{\mname{}}$ be the set of traces of this transition system.
    Consider the quantitative hyperproperty: $\qformulasub$ for this system where $Z \subset X$.
    We assume the property is well-defined.
    Further suppose:
    
    $\eqrel$ is the equivalence relation over traces corresponding to $\lnot\noteq$.
    \begin{align*}
        \trc{j} \eqrel \trc{k} \text{ iff } 
        \{ \trv{j} \mapsto \trc{j}, \trv{k} \mapsto \trc{k} \} 
            \models \lnot\noteq{}
    \end{align*}

    Let $\tracesetq{}(\trc{0})$ be the function defined as follows.
    \begin{align*}
        \tracesetq{}(\trc{0}) \doteq  
                \Big\{ \trc{1} ~|~ 
                    \{ \trv{0} \mapsto \trc{0}, \trv{1} \mapsto \trc{1} \}   
                        \models_{\traces{\mname{}}} \varphi_{\trc{0}, \trc{1}} \Big\}
    \end{align*}

    Let $\cntv{}(\trc{0})$ be the number of equivalence classes in $\tracesetq{}(\trc{0})$ induced by $\eqrel{}$.

    Separately, let $\tracesetcnt{}(\trc{0})$ be a function that constructs the maximally large set satisfying the following conditions:
    \begin{enumerate}
        \item $\forall \trc{j}, \trc{k} \in \tracesetcnt{}(\trc{0}).~ \trc{j} \neq \trc{k} \myiff \{ \trv{j} \mapsto \trc{j}, \trv{k} \mapsto \trc{k} \} \models\;\noteq{}$, and
        \item $\forall \trc{1} \in \tracesetcnt{}(\trc{0}).~ 
            \Big\{ \trv{0} \mapsto \trc{0}, \trv{1} \mapsto \trc{1} \Big\}
                \models_{\traces{\mname{}}} \varphi_{\trv{0}, \trv{1}}$.
    \end{enumerate}

    Then for every $\trc{0} \in \traces{\mname{}}$, $\cntv{}(\trc{0}) =\sizeof{\tracesetcnt{}(\trc{0})}$.
\end{lemma}

\begin{proof}
    For every set $\tracesetcnt{}(\trc{0})$ defined as above, every pair of distinct members $\trc{j}$ and $\trc{k}$, satisfy $\noteq$. 
    This means they are in different equivalence classes for the relation $\eqrel$.
    If $\tracesetcnt{}(\trc{0})$ is maximally large, then the number of equivalence classes must be equal to $\cntv{}(\trc{0})$.
\qed
\end{proof}

A consequence of Lemma~\ref{lem:eqclass} is that one can prove satisfiability of a Quantitative HyperLTL formula by counting the number of equivalence classes induced by $\eqrel$ over $\tracesetq{}(\trc{0})$ instead of using the definition shown in Figure~\ref{fig:sem-hyper-ltl}. 
\tclb*

\begin{proof}
    The proof is by induction on the number of satisfying assignments to $Y$ in the formula $\valid{Y, Z}{}$.
    We will consider the degenerate case when $\valid{Y, Z}{}$ is unsatisfiable separately and then use the case with one satisfying solution as the base case for the induction.

    If $\countsat{\valid{Y,Z}{}}{Y}$ is zero, then the lemma is trivially satisfied.

    {\bf Base case:} 
            If $\countsat{\valid{Y,Z}{}}{Y} = 1$, then by the definition of injective trace enumerations, for every trace $\trc{0}$, we have at least one trace $\trc{1}$ which satisfies $\varphi$.
            This means the number of equivalence classes for $\tracesetq{}(\trc{0})$ as defined in Lemma~\ref{lem:eqclass} is at least one. Thus the QHP is satisfied.

    \newcommand{\ysat}{\mathtt{y}}
    {\bf Induction hypothesis:} Suppose the lemma holds for all formulas $\mathcal{F}(Y, Z)$ such that $\countsat{\mathcal{F}(Y,Z)}{Y} = n$;
    $\mathcal{F}(Y, Z)$ and $\tren{}$ form an injective trace enumeration for the system $\mname{}$.
    By assumption, $\valid{}{}$ and $\tren{}$ also form an injective trace enumeration for $\mname{}$.
    
    {\bf Inductive step:} We now have to show that the lemma holds for an arbitrary formula $\valid{Y, Z}{}$ such that $\countsat{\valid{Y, Z}{}}{Y} = n + 1$ with $\valid{}{}$ and $\tren{}$ constituting an injective trace enumeration for the QHP on $\mname{}$.

    To do this, let us consider $\ysat$ which is a satisfying assignment to $\valid{Y, Z}{}$. 
    The formula $\validprime{Y}{} \doteq \valid{Y, Z}{} \land Y \neq \ysat{}$ has $n$ satisfying assignments to $Y$ for the same value of $Z$.
    Further, $\validprime{Y}{}$ and $\tren{}$ are also an injective trace enumeration for the QHP on the system $\mname{}$.
    By the induction hypothesis, the lemma holds for $\validprime{Y}{}$.
    This means $\mname \models \qformulavalidprime{\geq}$ and by Lemma~\ref{lem:eqclass} the set $\tracesetq{}(\trc{0})$ has at least $n$ equivalence classes for every $\trc{0} \in \traces{\mname{}}$, each class corresponding to the satisfying assignments to $Y$ in $\validprime{Y, Z}{}$ for $Z = \valueof{Z}{\trcij{0}{0}}$.
    If $\tracesetq{}(\trc{0})$ has more than $n$ equivalence classes, the lemma holds.

    Instead, suppose that $\tracesetq{}(\trc{0})$ has exactly $n$ equivalence classes.
    Now consider the assignment $\ysat{}$ and its corresponding trace $\trc{\ysat{}}$ in Property~\ref{eqn:inj-te-1}. By Property~\ref{eqn:inj-te-3}, $\trc{y}$ is in a different equivalence class from all of the traces corresponding to assignments to $Y$ in $\validprime{Y, Z}{}$. Contradiction!
    Therefore, $\tracesetq{}(\trc{0})$ has at least $n+1$ equivalence classes and so the inductive step holds.
    \qed
\end{proof}

\tcub*

\begin{proof}
    The proof is similar in structure to Lemma~\ref{lem:tclb} and is also by induction on the number of satisfying assignments to $Y$ in $\valid{Y, Z}{}$.

    As before, we will treat the degenerate case when $\countsat{\valid{Y, Z}{}}{Y} = 0$ separately.
    In this case, then Property~\ref{eqn:surj-te-1} must hold vacuously and there are no traces that satisfy $\varphi$. $\tracesetq{}(\trc{0})$ is the empty set for at least one $\trc{0} \in \traces{\mname{}}$. The QHP is satisfied.

    {\bf Base Case}: Since $\countsat{\valid{Y, Z}{}}{Y} = 1$, $\fixy{}_1\neq\fixy{}_2$ is false and so Property~\ref{eqn:surj-te-2} holds vacuously.
    This means there is only one equivalence class of traces, and so the lemma is true for the base case.

    {\bf Induction hypothesis:} As in the previous proof, suppose the lemma holds for all formulas $\mathcal{F}(Y, Z)$ such that $\countsat{\mathcal{F}(Y,Z)}{Y} = n$; 
    $\mathcal{F}(Y, Z)$ and $\tren{}$ form a surjective trace enumeration for the system $\mname{}$.
    By assumption, $\valid{}{}$ and $\tren{}$ form a surjective trace enumeration for $\mname{}$.
    
    \newcommand{\ysat}{\mathtt{y}}

    {\bf Inductive step:} We now have to show that the lemma holds for an arbitrary formula $\valid{Y, Z}{}$ such that $\countsat{\valid{Y, Z}{}}{Y} = n + 1$ when $\valid{}{}$ and $\tren{}$ constitute a surjective trace enumeration for the QHP on the transition system $\mname{}$.
    We will again consider $\ysat$, a satisfying assignment to $\valid{Y, Z}{}$ and construct the formula $\validprime{Y}{} \doteq \valid{Y, Z}{} \land Y \neq \ysat{}$.
    $\validprime{}{}$ has $n$ satisfying assignments to $Y$ for this value of $Z$.
    Consider the set $\traces{\mname{}'}$ constructed by removing all traces from $\traces{\mname{}}$ which are related to some trace $\trc{0}$ in the relation $\tren{}$ for $Y = \ysat{}$:
    $\traces{\mname{}'} = \traces{\mname{}} - \{ \trc{1} ~|~ \exists \trc{0}.~\tren{\ysat, \trc{0}, \trc{1}} \}$.
    For every $\tracesetq{}'(\trc{0})$ constructed analogously to $\tracesetq{}(\trc{0})$ from this set $\traces{\mname{}'}$, the set $\tracesetq{}'(\trc{0})$ has at least one less equivalence class (for the relation $\eqrel$) than $\tracesetq{}(\mname{})$.
    By the induction hypothesis, the lemma holds for the bound $\countsat{\validprime{Y, Z}{}}{Y}$ for set of traces $\traces{\mname{}'}$.
    However, $\tracesetq{}(\trc{0})$ has at most one more equivalence class than $\tracesetq{}'(\trc{0})$ and $\countsat{\valid{Y,Z}{}}{Y} = \countsat{\validprime{Y,Z}{}}{Y} + 1$.
    Therefore, the inductive step holds.
    \qed
\end{proof}

    \section{Experiments}
\label{appendix:microbenchmarks}
This section provides more detail on the benchmarks used in our experimental evaluation.  

    \paragraph{Electronic Purse:}
    This example models an electronic purse, which was also studied by Backes et al.~\cite{backes-09}, where a fixed amount $(\decr)$, modeled as an integer, is debited from the purse, with secret initial balance (also modeled as an integer), until the balance is insufficient for this transaction.
    The adversary-observable state consists of each debit and the number of debits from the purse.
    We show a deniability property which states that the number of traces with identical observations but different initial balances are at least $\decr{}$: 
   $ 
      \forall \trv{0}.\countq{\trv{1}}{\ltlf{(\delta{}_{\trv{j}, \trv{k}})}}{\ltlg{(\psi_{\trv{0}, \trv{1}})}}{\geq}{\decr}
    $.
    The binary predicate $\delta$ captures the fact that the traces have different balances, and $\psi$ ensures that they have identical adversary observations.

    
        
    \paragraph{Password Checker:}
    This example models a program that checks a password. 
    This example is interesting because Finkbeiner et al.~\cite{mcqhyper-2018} studied the same problem and found their technique times out when checking for 8-bit leakage for an 8-bit password. 
    We are able to provide a proof of quantitative non-interference (QNI) for a password of unbounded length (at the cost of manual construction of the trace enumeration relation).
    The QNI property states that the maximum information leaked to the attacker is $n$ bits, where $n$ is a state variable/parameter of the system representing the bit-length of the password: 
    $ 
      \forall \trv{0}.\countq{\trv{1}}{\ltlf{(\delta{}_{\trv{j}, \trv{k}})}}{\ltlg{(\psi_{\trv{0}, \trv{1}})}}{\leq}{2^n-1}
    $.
    $\delta$ requires the adversary observations to be different while $\psi$ ensures adversary inputs are the same in the two traces. 
    In other words, we prove that nothing besides the password is leaked to the attacker even if the adversary can make unlimited login attempts.

    \paragraph{Array Shuffle:}
    This implements a variant of the Fisher-Yates shuffle. 
    We chose this example because producing random permutations of an array is an important component of several cryptographic protocols (e.g., Ring ORAM~\cite{roram-15}) and a buggy shuffle algorithm that does not produce all permutations would result in vulnerabilities in these protocols.
    We prove a quantitative information flow property stating that all possible permutations are indeed generated by the algorithm: $
      \forall \trv{0}.\countq{\trv{1}}{\ltlf{(\delta{}_{\trv{j}, \trv{k}})}}{\ltlg{(\psi_{\trv{0}, \trv{1}})}}{\geq}{n!}$.
    As usual, $\delta$ requires the output of the shuffle to be different, while $\psi$ ensures that the input arrays are the same.
    This example is interesting because we are able to prove for unbounded-length input arrays that the shuffle does indeed produce all permutations.
    Note such an unbounded proof is not possible with techniques based solely on model counting.

\end{subappendices}
}
{
}

\end{document}